\documentclass[11pt]{amsart}

\usepackage[USenglish]{babel}

\usepackage{amsmath,amsthm,amssymb,amscd}
 
\usepackage{breqn}

\usepackage{longtable}
\usepackage{tikz}
\usetikzlibrary{arrows,shapes.geometric,positioning,decorations.markings}
 
\usepackage{url}

\usepackage[mathscr]{eucal}
\usepackage[normalem]{ulem}
\usepackage{latexsym,youngtab}
\usepackage{multirow}
\usepackage{epsfig}
\usepackage{parskip}
\usepackage{enumerate}

\usepackage[colorlinks=true]{hyperref}
 \hypersetup{
     colorlinks,
     citecolor={blue!50!black},
     linkcolor={red!50!black},
     urlcolor={green!30!black}
 }

\usepackage[top = 2.5cm, bottom = 2.5cm, left = 2.5cm, right = 2.5cm]{geometry}

\newcommand{\ootimes}{\otimes \cdots \otimes}
\newcommand{\ttimes}{\times \cdots \times}

\usepackage{booktabs}

\newcommand{\trace}{\mathrm{trace}}

\newcommand{\textsum}{{\textstyle \sum}}

\renewcommand{\theta}{\vartheta}

\newcommand{\Id}{\mathrm{Id}}
\newcommand{\rank}{\mathrm{rank}}

\newcommand{\RR}{\mathbb{R}}
\newcommand{\CC}{\mathbb{C}}
\newcommand{\QQ}{\mathbb{Q}}
\newcommand{\ZZ}{\mathbb{Z}}

\newcommand{\aaa}{\mathbf{a}}
\newcommand{\bbb}{\mathbf{b}}
\newcommand{\ccc}{\mathbf{c}}
\newcommand{\nnn}{\mathbf{n}}
\newcommand{\mmm}{\mathbf{m}}
\renewcommand{\lll}{\mathbf{l}}

\newcommand{\fg}{\mathfrak{g}}

\newcommand{\FS}{\mathfrak{S}}
\newcommand{\frakS}{\mathfrak{S}}

\newcommand{\calM}{\mathcal{M}}

\newcommand{\Mone}{M_{\langle 1\rangle}}

\newcommand{\ur}{\underline{\mathbf{R}}}
\newcommand{\uR}{\underline{\mathbf{R}}}
\newcommand{\uQ}{\underline{\mathbf{Q}}}

\newcommand{\supp}{\mathrm{supp}}
\newcommand{\aR}{\uwave{\mathbf{R}}}
\newcommand{\aQ}{\uwave{\mathbf{Q}}}

\newcommand{\ww}{\wedge}

\newcommand{\tdet}{\operatorname{det}}

\newcommand{\eps}{\varepsilon}

\newcommand{\link}{\url{http://fulges.github.io/code/CGLV/index.html}}

\usepackage{color}

\author[A. Conner, F. Gesmundo, J. M. Landsberg, E. Ventura]{Austin Conner, Fulvio Gesmundo, Joseph M. Landsberg, Emanuele Ventura}

\address[A. Conner]{Department of Mathematics, Harvard University, Cambridge, MA 02138, USA -- (previous) Department of Mathematics, Texas A\&M University, College Station, TX 77843-3368, USA}
\email[A. Conner]{connerad@math.tamu.edu}

\address[J. M. Landsberg]{Department of Mathematics, Texas A\&M University, College Station, TX 77843-3368, USA}
\email[J. M. Landsberg]{jml@math.tamu.edu}

\address[F. Gesmundo]{Max Planck Institute for Mathematics in the Sciences, Inselstrasse 22, 04103 Leipzig, Germany -- 
(previous) QMATH, Dept. Math. Sciences, University of Copenhagen, Universitetsparken 5, 2100 Copenhagen O., Denmark}
\email{fulvio.gesmundo@mis.mpg.de}

\address[E. Ventura]{Universit\`a di Torino, Dipartimento di Matematica, via Carlo Alberto 10, 10123 Torino, Italy -- (previous) Mathematisches Institut, Universit\"at Bern, Sidlerstrasse 5, 3012, Bern, Switzerland}
\email{emanuele.ventura@unito.it}

\title[Kronecker powers of tensors]{Rank and border rank of Kronecker powers of tensors and Strassen's laser method}



\renewcommand{\a}{\alpha}
\renewcommand{\b}{\beta}

\newcommand{\trank}{\mathrm{rank}}
\newcommand{\tdim}{\mathrm{dim}}

\newcommand{\ra}{\to}
\renewcommand{\bar}[1]{\overline{#1}}

\newcommand{\BC}{\mathbb{C}}

\newcommand{\vvirg}{, \dots ,}

\newcommand{\ot}{\otimes}

\newcommand{\Hom}{\mathrm{Hom}}
\newcommand{\perm}{\mathrm{perm}}
\newcommand{\Mat}{\mathrm{Mat}}
\renewcommand{\det}{\mathrm{det}}

\newcommand{\bbC}{\mathbb{C}}
\newcommand{\bbN}{\mathbb{N}}

\newcommand{\bbQ}{\mathbb{Q}}

\newcommand{\bfb}{\mathbf{b}}

\newcommand{\bfl}{\mathbf{l}}
\newcommand{\bfm}{\mathbf{m}}
\newcommand{\bfn}{\mathbf{n}}

\newcommand{\calS}{\mathcal{S}}
\newcommand{\calT}{\mathcal{T}}

\newcommand{\bfR}{\mathbf{R}}

\newcommand{\frakgl}{\mathfrak{gl}}
\newcommand{\frakg}{\mathfrak{g}}

\newcommand{\triv}{\mathrm{triv}}

\newtheorem{theorem}{Theorem}[section]
\newtheorem{corollary}[theorem]{Corollary}

\newtheorem{lemma}[theorem]{Lemma}
\newtheorem{proposition}[theorem]{Proposition}
\newtheorem{theorem*}[theorem]{Theorem*}
\newtheorem{problem}[theorem]{Problem}

\theoremstyle{definition}
\newtheorem{remark}[theorem]{Remark}
\newtheorem{example}[theorem]{Example}

\newcommand{\bbT}{\mathbb{T}}
\newcommand{\bbZ}{\mathbb{Z}}\newcommand{\BZ}{\mathbb{Z}}
\newcommand{\tperm}{\mathrm{perm}}

\begin{document}
\begin{abstract} 

We prove that the border rank of the Kronecker square of the little Coppersmith-Winograd tensor $T_{cw,q}$ is the square of its border rank for $q > 2$ and that the border rank of its Kronecker cube is the cube of its border rank for $q > 4$. This answers questions raised implicitly in \cite[\S 11]{copwin135} and explicitly in \cite[Problem 9.8]{gs005} and rules out the possibility of proving new upper bounds on the exponent of matrix multiplication using the square or cube of a little Coppersmith-Winograd tensor in this range.

In the positive direction, we enlarge the list of explicit tensors potentially useful for Strassen's laser method, introducing a skew-symmetric version of the Coppersmith-Winograd tensor, $T_{skewcw,q}$. For $q = 2$, the Kronecker square of this tensor coincides with the $3\times 3$ determinant polynomial, $\det_3 \in \BC^9\ot \BC^9\ot \BC^9$, regarded as a tensor. We show that this tensor could potentially be used to show that the exponent of matrix multiplication is two.

We determine new upper bounds for the (Waring) rank and the (Waring) border rank of $\det_3$, exhibiting a strict submultiplicative behaviour for $T_{skewcw,2}$ which is promising for the laser method. 

We establish general results regarding border ranks of Kronecker powers of tensors, and make a detailed study of Kronecker squares of tensors in $\BC^3\ot \BC^3\ot \BC^3$.
\end{abstract}
\maketitle

\section{Introduction}

The exponent $\omega$ of matrix multiplication is defined as
\[
\omega :=\inf \{\tau \mid  \text{ two } \nnn \times \nnn \text{ matrices may be multiplied using } O(\nnn^\tau) \text{ arithmetic operations}\}.
 \]
This is a fundamental constant governing the complexity of the basic operations in linear algebra. It is conjectured that $\omega = 2$. There is a classical upper bound $\omega \leq 3$ following from the standard row-by-column multiplication. Starting from 1969 \cite{Strassen493}, a great deal of effort has been spent on the research on upper bounds on the exponent, involving methods from combinatorics, probability, and statistical mechanics; we refer to Section \ref{hist} for a brief history. The more recent Cohn-Umans approach \cite{CU} uses group-theoretic techniques and in particular the Fourier-transform of finite groups. In this work, we approach the problem via algebraic geometry and representation theory. We obtain both negative and hopeful results.

Our focus will be on Strassen's \emph{laser method} \cite{MR882307}. This technique was used to achieve Strassen's upper bound of 1988 and essentially all subsequent upper bounds. In order to present the method and our contributions, we adopt the language of tensors.

\subsection{Definitions}
Let $A,B,C$ be complex vector spaces. A tensor $T\in  A \ot B\ot C$ has {\it rank one} if $T=a\ot b\ot c$ for some $a\in A$, $b\in B$, $c\in C$. The {\it rank} of $T$, denoted $\bfR(T)$, is the  smallest $r$ such that $T$ is sum of $r$ rank-one tensors.  The {\it border rank} of $T$, denoted $\ur(T)$, is the smallest $r$ such that $T$ is the limit of a sequence of rank $r$ tensors.

A tensor $T \in A \otimes B \otimes C$ defines a bilinear map $A^* \times B^* \to C$ and a trilinear map $A^* \times B^* \times C^* \to \bbC$. The matrix multiplication tensor $M_{\bfl,\bfm,\bfn}$ is the tensor associated to the bilinear map
\[
 M_{\langle \bfl,\bfm,\bfn \rangle} : \Mat_{\bfl \times \bfm} \times \Mat_{\bfm \times \bfn} \to \Mat_{\bfl \times \bfn}
\]
sending a pair of matrices $(X,Y)$ to their product $XY$. As a trilinear map, the matrix multiplication tensor is $M_{\langle \lll,\mmm,\nnn \rangle}(X,Y,Z)=\trace(XYZ)$, where $X,Y,Z$ are matrices of size $\lll\times \mmm$, $\mmm\times \nnn$ and $\nnn\times\lll$, respectively. The matrix multiplication tensor has the following important self-reproducing property: $M_{\langle \lll,\mmm,\nnn \rangle} \boxtimes M_{\langle \lll',\mmm'\nnn' \rangle} = M_{\langle \lll\lll',\mmm\mmm',\nnn\nnn'\rangle}$. Write $M_{\langle \bfn \rangle} :=M_{\langle \bfn , \bfn, \bfn\rangle}$.

The complexity of performing a bilinear map, and in particular the complexity of matrix multiplication, is controlled by the tensor rank of the corresponding tensor. Bini \cite{MR605920} showed that border rank controls the complexity as well:
\[
\omega = \inf \{ \tau : \uR(M_{\langle \bfn \rangle}) \in O(\bfn ^\tau)\}. 
\]

Let $GL(A)$ be the general linear group of invertible linear maps $A \to A$ and similarly for $B$ and $C$.  We say that two tensors are \emph{isomorphic} if they are in the same orbit under the 
natural action of $GL(A) \times GL(B) \times GL(C)$ on $A\ot B\ot C$. We will often assume that all tensors involved in the discussion belong to the same space $A \otimes B \otimes C$. This is not restrictive, since we may re-embed the spaces $A,B,C$ into larger spaces whenever it is needed.

Given $T,T' \in A \otimes B \otimes C$, we say that $T$ \emph{degenerates} to $T'$ if $T' \in \bar{GL(A) \times GL(B) \times GL(C) \cdot T}$, the closure of the orbit of $T$, equivalently in the Euclidean or in the Zariski topology. Border rank is semicontinuous under degeneration: $\uR(T' ) \leq \uR(T)$ if $T$ degenerates to $T'$.

Border rank may be rephrased in terms of degeneration as follows. For a tensor $T$, one has $\uR(T) \leq r$ if and only if $T$ is a degeneration of $\Mone^{\oplus r} = \sum_{i=1}^r a_i \otimes b_i \otimes c_i$, where $\{a_i\}$ is a set of linearly independent vectors and similarly for $\{b_i\}$ and $\{c_i\}$. The \emph{border subrank} of $T$, denoted $\uQ(T)$, is the maximum $q$ such that $T$ degenerates to $\Mone^{\oplus q}$. 

For tensors $T\in A\ot B\ot C$ and $T'\in A'\ot B'\ot C'$, the {\it Kronecker product} of $T$ and $T'$ is the tensor $T\boxtimes T' := T \ot T' \in (A\ot A')\ot (B\ot B')\ot (C\ot C')$, regarded as $3$-way tensor. Given $T \in A \otimes B \otimes C$, the {\it Kronecker powers}  of $T$ are $T^{\boxtimes N} \in A^{\otimes N} \otimes B^{\otimes N} \otimes C^{\otimes N}$, defined iteratively. Rank and border rank are submultiplicative under Kronecker product: $\bfR(T \boxtimes T') \leq \bfR(T) \bfR(T')$, $\uR(T \boxtimes T') \leq \uR(T) \uR(T')$, and both inequalities may be strict.  

Asymptotic versions of border rank and border subrank,  respectively called \emph{asymptotic rank} and \emph{asymptotic subrank},  are defined as follows:
\[
\aR(T) = \lim_{N \to \infty} [\uR(T^{\boxtimes N})] ^{1/N}, \quad \aQ(T) = \lim_{N \to \infty} [\uQ(T^{\boxtimes N})] ^{1/N} .
\]

One has $\omega = \log_2 (\aR(M_{\langle 2 \rangle}) $; in particular $\omega = 2$ if and only if $\aR(M_{\langle \bfn \rangle}) = \bfn^2$ for any (and as a consequence all) $\bfn$. 

\subsection{Strassen's laser method and its barriers}
The two fundamental ingredients of Strassen's laser method are submultiplicativity of border rank under Kronecker powers and semicontinuity of border rank under degeneration. The laser method relies on an auxiliary tensor $T$ with the property that $\uR(T)$ is \emph{small} and, for some large $N$, $T^{\boxtimes N}$ degenerates to a \emph{large} matrix multiplication tensor. 

Since 1987, only three  tensors have been employed in the method and the     best upper bounds so far 
come from the \emph{big Coppersmith-Winograd  tensor} \cite{copwin135}: 
\begin{align*}\label{bigcw} 
T_{CW,q}:= &\sum_{j=1}^q
a_0\ot b_j\ot c_j+a_j\ot b_0\ot c_j + a_j\ot b_j\ot c_0+\\ 
&a_0\ot b_0\ot c_{q+1}+ a_0\ot b_{q+1}\ot c_0+a_{q+1}\ot b_0\ot c_0 
\in (\BC^{q+2})^{\ot 3},
\end{align*}
It was used to prove $\omega < 2.38$ in 1988 and all further improvements to the current best known upper bound $\omega < 2.373$. 

In 2014 \cite{MR3388238} gave an explanation for the limited progress since 1988, followed by further explanations in \cite{DBLP:conf/innovations/AlmanW18,2018arXiv181008671A, DBLP:journals/corr/abs-1812-06952,DBLP:journals/corr/abs-1812-08731}. One major consequence of these results is that $T_{CW,q}$ cannot be used to prove $\omega <2.3$ using the standard laser method. 

A geometric identification of the barrier of \cite{MR3388238} was given in \cite{DBLP:journals/corr/abs-1812-06952}. Strassen showed $\uQ( M_{\langle \bfn \rangle} ) \geq  \lceil \frac{3}{4} n^2 \rceil$ (in \cite[Theorem 3]{KopMosZui:GeomRankSubrankMaMu} equality was proved). This, together with the self-reproducing property of the matrix multiplication tensor, implies $\aQ(M_{\langle \bfn \rangle} ) = \bfn^2$, which is the maximum possible value. A consequence is that no tensor having non-maximal asymptotic subrank can be used to prove $\omega=2$ via the laser method; in \cite{MR1089800} it was shown that $\aQ(T_{CW,q})$  is non-maximal. 

The second most effective tensor used for upper bounds via Strassen's laser method is the \emph{small Coppersmith-Winograd  tensor}:
\begin{equation}\label{littlecw}
T_{cw,q}:=\sum_{j=1}^q
a_0\ot b_j\ot c_j+a_j\ot b_0\ot c_j + a_j\ot b_j\ot c_0\in (\BC^{q+1})^{\ot 3}.
\end{equation}

In \cite{copwin135}, the laser method was used to (implicitly) prove the following result:
\begin{theorem}[Coppersmith-Winograd \cite{copwin135}]\label{cwbndk} 
 For all $k$ and $q$,
\begin{equation}\label{cwbndq}
\omega\leq  \log_q(\frac{4}{27}(\uR(T_{cw,q}^{\boxtimes k}))^{\frac{3}{k}}) .
\end{equation}  
\end{theorem}

One has $\ur(T_{cw,q})=q+2$, which is one more than minimal (see, e.g., \cite[Sec. 15.8]{BCS}). Applying Theorem \ref{cwbndk} to $T_{cw,8}$ with $k=1$ one obtains $\omega \leq 2.41$ \cite{copwin135}.  Theorem  \ref{cwbndk} implies that if the border rank of the  Kronecker square or some higher Kronecker power of $T_{cw,q}$ were strictly submultiplicative, one could get a better bound, and one could even potentially prove $\omega =2$ using Kronecker powers of  $T_{cw,2}$. Indeed, \cite[Ex. 15.24]{BCS} observes that Theorem \ref{cwbndk} holds replacing $\uR(T_{cw,q}^{\boxtimes k})^{\frac{1}{k}}$ with $\aR(T_{cw,q})$. In particular, were $\aR(T_{cw,2}) = 3$, then Theorem \ref{cwbndk} would imply $\omega =2$. This shows that the barriers of \cite{MR3388238,DBLP:conf/innovations/AlmanW18,2018arXiv181008671A, DBLP:journals/corr/abs-1812-06952} do not apply to $T_{cw,2}$. Previous to our work, the possibility to prove the upper bound $\omega < 2.3$ using the second and third Kronecker power of $T_{cw,q}$ for $3 \leq q \leq 10$ was open, in the sense that the if the state of art lower bound on $T_{cw,q}^{\boxtimes k}$ were equal to an upper bound, then Theorem \ref{cwbndk} would have given an improvement. We show that this is not the case.

\subsection{Main results}

M. Bl\"aser  \cite[Problem 9.8]{gs005}  posed the problem of determining the border rank of $T_{cw,q}^{\boxtimes 2}$. We show:
\begin{theorem}\label{thm: lower bounds cw powers}~\\
For all $q>2$, $\ur(T_{cw,q}^{\boxtimes 2})=(q+2)^2$; moreover $15  \leq \ur(T_{cw,2}^{\boxtimes 2})\leq 16$.

For all $q > 4$, $\uR(T_{cw,q}^{\boxtimes 3}) = (q+2)^3$; if $q = 3,4$ then $\ur(T_{cw,q}^{\boxtimes 3})\geq (q+2)^2 (q+1)$; if $q=2$, then $\ur(T_{cw,2}^{\boxtimes 3})\geq 15 \cdot 3$.

For all $q>4$ and all $N$, $\ur(T_{cw,q}^{\boxtimes N})\geq (q+1)^{N-3}(q+2)^3$; if $q = 3,4$ then $\ur(T_{cw,q}^{\boxtimes N})\geq (q+2)^2 \cdot (q+1)^{N-2}$; if $q = 2$ then $\ur(T_{cw,2}^{\boxtimes 3})\geq 15 \cdot 3^{N-2}$.
\end{theorem}
This improves on the previous lower bound from \cite{MR3578455}, which was $\ur(T_{cw,q}^{\boxtimes N})\geq (q+1)^N + 2^N-1$ for all $q,N$.

This result shows that the second and third Kronecker powers of $T_{cw,q}$ cannot give any improvement on the current upper bounds on the exponent. For instance, the lower bound of \cite{MR3578455} for $(q,N) = (3,3)$ is $\ur(T_{cw,3}^{\boxtimes 3}) \geq 71$; if this had been the value of $\uR(T_{cw,3}^{\boxtimes 3})$ then Theorem \ref{cwbndk} would have given $\omega < 2.15$; however, the lower bound of Theorem \ref{thm: lower bounds cw powers} guarantees $\uR(T_{cw,3}^{\boxtimes 3}) \geq 100$, and even if this turns out to be the value of $\uR(T_{cw,3}^{\boxtimes 3})$, Theorem \ref{thm: lower bounds cw powers} only gives $\omega < 2.46$.

In light of the above-mentioned barriers and Theorem \ref{thm: lower bounds cw powers}, one might try to determine better tensors which are not subject  to the barriers  (similarly to $T_{cw,q}$) and at the same time have strict submultiplicativity of border rank under Kronecker powers.

Inspired by \cite{2019arXiv190909518C}, we introduce a new family of tensors, which are a skew-symmetric version of the small Coppersmith-Winograd tensors for every even $q$:
\begin{equation}\label{littleskcw} 
T_{skewcw,q}:=\sum_{j=1}^q
a_0\ot b_j\ot c_j+a_j\ot b_0\ot c_j + \sum_{\xi =1}^{\frac{q}{2}} (a_\xi\ot b_{\xi+\frac{q}{2}}- a_{\xi+\frac{q}{2}}\ot b_\xi)\ot c_0\in (\BC^{q+1})^{\ot 3}. 
\end{equation}
Proposition \ref{skewcwbnd} shows Theorem \ref{cwbndk} holds with $T_{cw,q}$ is replaced by $T_{skewcw,q}$, so in particular $T_{skewcw,2}$ could potentially be used to prove $\omega = 2$. 

 Proposition \ref{skewbad} contains more negative news: $\uR(T_{skewcw,q}) \geq q+3$, and in particular $\uR(T_{skewcw,2}) = 5$. However, we show a strong submultiplicative behaviour for $T_{skewcw,q}$, namely $\uR(T_{skewcw,2}^{\boxtimes 2}) \leq 17<5^2$.   Theorem \ref{thm: waring rank and border rank det3} below actually proves a stronger statement. We show in Lemma \ref{pdlem} that $T_{skewcw,2}^{\boxtimes 2}$ is isomorphic to the $3 \times 3$ determinant polynomial regarded as a tensor and we prove new upper bounds for the symmetric rank (also known as Waring rank, see, e.g., \cite[\S2.6.6]{MR2865915}) and symmetric border rank of the $3 \times 3$ determinant polynomial.
 
\begin{theorem} \label{thm: waring rank and border rank det3}
Let $\det_3 \in \bbC^9 \otimes \bbC^9 \otimes \bbC^9$ be the $3 \times 3$ determinant polynomial regarded as a symmetric tensor. Then 
\[
\bfR_S(\det_3) \leq 18, \quad  \text{and} \quad \uR_S(\det_3) \leq 17. 
\]
\end{theorem}
In \cite{CHLbapolar}, it was shown that $\uR(\det_3) = 17$ and in particular the second inequality in Theorem \ref{thm: waring rank and border rank det3} is an equality.

The proof of Theorem \ref{thm: lower bounds cw powers} is given in Section \ref{kossect} and the proof of Theorem \ref{thm: waring rank and border rank det3} is given in Section \ref{section: determinant upper bounds}.

Some of the proofs of this work rely on computer calculations performed by the software Macaulay2 \cite{M2} and Sage \cite{sagemath}. The scripts performing these calculations are collected in different appendices in the Supplementary Material available at 
\begin{center}
 \link
\end{center}

\subsection{Brief history of upper bounds}\label{hist}
There was steady progress in the research for upper bounds on $\omega$  from 1968 to 1988. 

In \cite{Strassen493}, Strassen gave an algorithm to perform $2\times 2$ matrix multiplication using $7$ scalar multiplications (rather than the $8$ multiplications of the standard algorithm). This gives the upper bound $\omega< 2.81$. In \cite{MR592760}, Bini et. al., using border rank,  showed $\omega<2.78$. 

A major breakthrough due to Sch\"onhage \cite{MR623057}, known as the asymptotic sum inequality, was used to show $\omega<2.55$ by exploiting the interplay between direct sums and the self-reproducing property of the matrix multiplication tensor. In \cite{MR882307} Strassen introduced the \emph{laser method} and showed $\omega<2.48$. A refined form of the laser method was used by Coppersmith and Winograd to show $\omega<2.3755$ \cite{copwin135}.

There was no progress on upper bounds on the exponent until 2011 when, via a further refinement of the method, a series of improvements by Stothers, Williams, Le Gall and Alman and Williams and \cite{stothers,williams,LeGall:2014:PTF:2608628.2608664,AlmWil:RefinedLaserMethodFMM} lowered the upper bound to the current state of the art  $\omega < 2.373$.

\section{Preliminary results}

In this section, we provide some results which will be useful in the rest of the paper.

The following remark shows that submultiplicativity holds asymptotically for most tensors; this appeared implicitly in \cite[Lemma 3.5]{MR929980} and then explicitly in \cite[Prop. 2.12]{DBLP:journals/corr/abs-1812-06952}. 

\begin{remark} \label{prevkp}
Any $T\in \BC^m\ot \BC^m\ot \BC^m$ is a degeneration of $M_{\langle 1,m,m\rangle}\in \BC^{m^2}\ot \BC^m\ot \BC^m$, so $T^{\boxtimes 3}$ is a degeneration of $M_{\langle m^2,m^2,m^2\rangle}= M_{\langle 1,m,m\rangle}\boxtimes  M_{\langle m,1, m\rangle}\boxtimes M_{\langle  m,m,1\rangle}$. In particular $\aR(T^{\boxtimes 3})\leq \aR(M_{\langle m^2,m^2,m^2\rangle}) = m^{2\omega}$ and therefore $\aR(T) \leq m^{2\omega /3}$. Since $\omega < 2.4$, we have $\aR(T) \leq m^{1.6}$. 
\end{remark}

\subsection{$T_{skewcw,q}$ and the laser method}
The first result is the analog of Theorem \ref{cwbndk} for the family $T_{skewcw,q}$: 

\begin{proposition}   \label{skewcwbnd} 
 For all $k$,
\begin{equation}\label{skewcwbndq}
\omega\leq  \log_q(\frac{4}{27}(\uR(T_{skewcw,q}^{\boxtimes k}))^{\frac{3}{k}}) .
\end{equation}  
\end{proposition}
\begin{proof}
Similarly to the case of $T_{cw,q}$, the proof follows immediately from \cite[Theorem 15.41]{BCS}, because $T_{skewcw,q}$ has the same \lq\lq block structure\rq\rq\  as $T_{cw,q}$.
 \end{proof}
 
In particular, similarly to $T_{cw,q}$, if $\aR(T_{skewcw,2}) = 3$ then $\omega = 2$ and it is potentially possible to improve the current upper bounds on $\omega$ using $T_{skewcw,q}$. Therefore, it is important to determine upper bounds on the border rank of the Kronecker powers of $T_{skewcw,q}$, and in particular in the case $q=2$.

\subsection{Coppersmith-Winograd tensors, symmetries, determinants and permanents}\label{subsec: symmetry group of tensors}

Let  $S^3 \bbC^m$ and   $\Lambda^3 \bbC^m$ 
respectively denote the subspaces of symmetric and skew-symmetric tensors in $\bbC^m \otimes \bbC^m \otimes \bbC^m$. By identifying the three copies of $\bbC^{q+1}$ in \eqref{littlecw} and \eqref{littleskcw}, we observe that $T_{cw,q}$ is isomorphic to a symmetric tensor and $T_{skewcw,q}$ is isomorphic to a skew-symmetric tensor. Inded, fixing a basis $a_0 \vvirg a_q$ of $\bbC^{q+1}$, the isomorphism $a_j \leftrightarrow b_j \leftrightarrow c_j$ provides 
\begin{equation}\label{eqn: cw as sym}
 T_{cw,q} = a_0 (a_1^2 + \cdots + a_q^2) \in S^3 \bbC^{q+1}.
\end{equation}
Similarly, if $q = 2u$ is even, the isomorphism
\begin{align*}
& a_0 \leftrightarrow -b_0 \leftrightarrow c_0\\
& a_j \leftrightarrow b_j \leftrightarrow -c_{u+j} \quad j = 1 \vvirg u \\
& a_{u+j} \leftrightarrow b_{u+j} \leftrightarrow c_{j} \quad j = 1 \vvirg u \\
\end{align*}
provides
\begin{equation}\label{eqn: skewcw as skew}
 T_{skewcw,q} = a_0 \wedge (a_1 \wedge a_{u+1} + \cdots + a_u \wedge a_q) \in \Lambda^3 \bbC^{q+1}. 
\end{equation}

We introduce some definitions concerning the symmetries of a tensor. The  group homomorphism $\Phi: GL(A) \times GL(B) \times GL(C) \to GL(A \otimes B \otimes C)$ defining the natural action on $A\ot B\ot C$ has a two dimensional kernel $\ker   \Phi = \{ (\lambda \Id_A, \mu\Id_B, \nu \Id_C): \lambda \mu \nu = 1\} \simeq (\bbC^*)^2$. 

In particular, the group $ \left( GL(A) \times GL(B) \times GL(C) \right) /(\bbC^*)^{\times 2}$ is identified with a subgroup of $GL(A \otimes B \otimes C)$. Given $T\in A\ot B\ot C$, the {\it symmetry group} of a tensor $T$ is the stabilizer of $T$ in $ \left(GL(A)\times GL(B)\times  GL(C)\right)/(\BC^*)^{\times 2} $, that is 
\begin{equation}\label{gtdef}
G_T:=\{ g\in  \left(GL(A)\times GL(B)\times GL(C)\right)/(\BC^*)^{\times 2}  
\mid \ g\cdot T=T\}.
\end{equation}

If the three spaces $A,B,C$ are identified, so that $A \otimes B \otimes C \simeq A^{\otimes 3}$, one can consider the action restricted to $GL(A)$ embedded diagonally as $GL^{diag}(A) \subseteq GL(A)^{\times 3}$. In this case, the kernel of the action reduces to the cyclic group $\bbZ_3 = \{ \zeta \Id_A : \zeta^3 =1\}$ and one can consider a restricted version of the symmetry group
\[
 G^s_T:= G_T \cap GL^{diag}(A) =\{g\in GL(A) / \bbZ_3 \mid g\cdot T=T\} .
\]

Let $\frakS_k$ be the permutation group on $k$ elements. 

We record the following observation:

\begin{proposition}\label{prop: symmetry group under kron powers}
Let $T\in A\ot B\ot C$ (resp. $T \in A^{\otimes 3}$). Then 
\[
G_{T^{\boxtimes N}} \supseteq G_T ^{\times N} \rtimes \frakS_N \quad (\text{resp. } G_{T^{\boxtimes N}}^s \supseteq {G_T^s} ^{\times N} \rtimes \frakS_N )
\]
where the symmetric group acts by permuting the factors of the direct product.
\end{proposition}
\begin{proof}
Let $T\in A\ot B\ot C$. Every factor $G_T$ in $ G_T ^{\times N} \rtimes \frakS_N$ acts on a single factor of $T^{\boxtimes N}$ and it stabilizes it by definition of $G_T$. The groups $\frakS_N$ permutes the factors of $T^{\boxtimes N}$, which is a Kronecker power and therefore it is stabilized.

The statement for $T \in A^{\otimes 3}$ is an immediate consequence.
\end{proof}

Consider the action of the symmetric group $\frakS_3$ which permutes the tensor factors. A tensor is symmetric if it is invariant under this action and skew-symmetric if it is skew-invariant. It is easy to observe that Kronecker powers of symmetric tensors are symmetric tensors. Moreover, odd Kronecker powers of skew-symmetric tensors are skew-symmetric and even Kronecker powers of skew-symmetric tensors are symmetric. 

We record the expressions of the $3 \times 3$ permanent and determinant polynomials as tensors in $\bbC^9 \otimes \bbC^9 \otimes \bbC^9$. Write $(-1)^\sigma$ for the sign of a permutation $\sigma$. Then
\begin{align*}
 \det_3 &= \frac{1}{6}\sum_{\sigma,\tau \in \frakS_3} (-1)^{\sigma\tau} a_{\sigma(1)\tau(1)} \otimes b_{\sigma(2) \tau(2)} \otimes c_{\sigma(3)\tau(3)}, \\ 
 \perm_3 &= \frac{1}{6} \sum_{\sigma,\tau \in \frakS_3} a_{\sigma(1)\tau(1)} \otimes b_{\sigma(2) \tau(2)} \otimes c_{\sigma(3)\tau(3)}. \\ 
\end{align*}
 
 \begin{lemma}\label{pdlem}
We have the following isomorphisms of tensors:  
 \begin{align*}
& T_{cw,2}^{\boxtimes 2} \cong \perm_3, \\ 
& T_{skewcw,2}^{\boxtimes 2} \cong  \det_3.
 \end{align*}
\end{lemma}
\begin{proof}

From \eqref{eqn: cw as sym}, we have $T_{cw,2} = a_0(a_1^2 + a_2^2)$. Let $\tilde{a}_1 = (a_1+ \sqrt{-1}a_2)$ and $\tilde{a}_2 = (a_1- \sqrt{-1}a_2)$, so that $T_{cw,2} = a_0 \tilde{a}_1 \tilde{a}_2$. This shows that after a suitable change of basis $T_{cw,2} = a_0a_1a_2$. Its symmetry group is $G^s_{T_{cw,2}} = \bbT_3^{SL} \rtimes \frakS_3$, where $\bbT^{SL}_3$ denotes the torus of diagonal matrices with determinant one, and $\frakS_3$   acts permuting the three basis elements. 

By Proposition \ref{prop: symmetry group under kron powers}, we deduce that $T_{cw,2}^{\boxtimes 2}$ is a symmetric tensor, with $G^s_{T_{cw,2}^{\boxtimes 2}} \supseteq (\bbT^{SL}_3 \rtimes \frakS_3)^{\times 2}\rtimes \BZ_2$ (and in fact equality holds). This is the stabilizer of the permanent polynomial $\perm_3$. Since the permanent is characterized by its stabilizer (see, e.g., Lemma \ref{lemma: det perm characterized} below), we conclude.

The proof for $T_{skewcw,2}$ is similar. From \eqref{eqn: skewcw as skew}, we have $T_{skewcw,2} = a_0 \wedge a_1 \wedge a_2$. Therefore $G^s_{T_{skewcw,2}} = SL_3$; indeed $T_{skewcw,2}$ is the unique, up to scale, $SL_3$-invariant in $\bbC^3 \otimes \bbC^3 \otimes \bbC^3$. 

By Proposition \ref{prop: symmetry group under kron powers}, we deduce that $T_{skewcw,2}^{\boxtimes 2}$ is a symmetric tensor, with $G^s_{T_{skewcw,2}^{\boxtimes 2}} \supseteq (SL_3)^{\times 2}\rtimes \BZ_2$ (and in fact equality holds). This is the stabilizer of the determinant polynomial $\det_3$. Since the determinant is characterized by its stabilizer, we conclude.
\end{proof}

The symmetric tensors $\det_m$ and $\perm_m$ are characterized by their stabilizers. For the determinant, this fact is classical. For the permanent, the statement, but not the proof, appears in \cite{GCT2}. For completeness, we provide a proof here assuming some familiarity with the representation theory of $SL_m$ and of the symmetric group $\frakS_m$.
\begin{lemma}\label{lemma: det perm characterized}
 Let $T \in S^m (\bbC^m \otimes \bbC^m)$ be a symmetric tensor of order $m$. If $G^s_T \supseteq (\bbT^{SL} \rtimes \frakS_m)^{\times 2} \rtimes \bbZ_2$, then $T = \perm_m$, up to scale. If $G^s_T \supseteq (SL_m^{\times 2}) \rtimes \bbZ_2$, then $T = \det_m$, up to scale.
\end{lemma}
\begin{proof}
First consider the case of the determinant. Let $SL_m\times SL_m=SL(E)\times SL(F)$ act on $S^m(E\ot F)$. This space decomposes as $SL(E) \times SL(F)$-representation as (see, e.g., \cite[\S6.7.6]{MR2865915})
$$
S^m(E\ot F) = \bigoplus_{|\pi|=m} S_{\pi}E\ot S_{\pi}F;
$$
this is multiplicity free, with the only trivial module $S_{(1^m)}E\ot S_{(1^m)}F=\Lambda^{m} E \otimes \Lambda^{m} F$. This is the space spanned by $\det_m$.

In the case of the permanent, note that the decomposition above holds for the action of $(\bbT^{SL} \rtimes \frakS_m)^{\times 2}$ as well. Then, the 
$ \bbT^{SL(E)}\times \bbT^{SL(F)}$-invariant subspace is given by the sum of the weight zero spaces $(S_{\pi}E)_0\ot (S_{\pi}F)_0$. By \cite{MR0414794}, 
one has the isomorphism $(S_{\pi}E)_0\ot (S_{\pi}F)_0=[\pi]_E\ot [\pi]_F$ for the weight zero spaces as $\frakS_E \times \frakS_F$-modules. The only trivial representation is the one corresponding to $\pi = (d)$, which is the subspace spanned by $\perm_m$. 
\end{proof}

  Lemma \ref{pdlem} guarantees that $\perm_3$ and $\det_3$ are tensors not subject to barriers for the laser method. In particular, either  $\aR(\det_3)=9$ or $\aR(\tperm_3)=9$ would imply $\omega=2$.

\begin{remark}
 A similar result holds for higher Kronecker powers. For every $k$, the even power $T_{skewcw,2}^{\boxtimes 2k}$ is invariant under $SL_3^{\times 2k} \rtimes \frakS_{2k}$. There is a unique invariant $\mathrm{PasDet}_{k,3}$ for $SL_3^{\times 2k}$ in $S^3 ((\BC^3)^{\ot 2k})$: it is the generator of the submodules $(\Lambda^3\bbC^3)^{\ot 2k}$, known as the \emph{Pascal determinant} (see, e.g., \cite[\S 8.3]{MR2865915}). If any of the Pascal determinants has minimal asymptotic rank, i.e., $\aR(\mathrm{PasDet}_{k,3}) = 3^{2k}$, then $\omega = 2$.
\end{remark}
  
\begin{remark}
 One can regard the $3 \times 3$ determinant and permanent as trilinear maps 
$\bbC^3 \times \bbC^3 \times \bbC^3\ra \BC$, where the three copies of $\bbC^3$ 
are the first, second and third column of a $3 \times 3$ matrix. From this 
point of view, the trilinear map given by the determinant is $T_{skewcw,2}$ as 
a tensor and the one given by the permanent is $T_{cw,2}$ as a tensor. This 
perspective, combined with the notion of product rank (in the sense of \cite{MR3492642}) provides the upper bounds $\uR_S(\perm_3) \leq 16$ and $\uR(\det_3) \leq 20$. These bounds already appeared in \cite{MR3494510,MR3492642} and are also a consequence of Lemma \ref{pdlem}.
\end{remark}

\subsection{Generic tensors in $\BC^3\ot \BC^3\ot \BC^3$}\label{genc3}
It is a classical fact that a generic tensor in $\BC^3\ot \BC^3\ot \BC^3$ has border rank five \cite{Strassen:RanksOptimalComputation}.

\begin{remark}  \label{numersquare}
Computer experiments indicate that for all $T\in \BC^3\ot \BC^3\ot \BC^3$, $\ur(T^{\boxtimes 2})\leq 22<25$.

Evidence for the remark is obtained as follows. We considered tensors $T \in \bbC^3 \otimes \bbC^3 \otimes \bbC^3$ whose coefficients in a fixed basis were taken independently and uniformly random in $[-1,1]$. We obtained numerically that $\uR(T^{\boxtimes 2}) \leq 22$. An instance of this computation is available in Appendix A of the Supplementary Material.
\end{remark}

\begin{problem} Prove the claim in Remark \ref{numersquare}. Even better, give a geometric proof.
\end{problem}

Remark \ref{numersquare} is not too surprising because $\BC^3\ot \BC^3\ot \BC^3$ is {\it secant defective}, in the sense that by a dimension count, one would expect the maximum border rank of a tensor to be $4$, but the actual maximum is $5$. This means that for a generic tensor, there is a $8$ parameter family of rank $5$ decompositions, and it is not surprising that the na{\"\i}ve $64$-parameter family of decompositions of the square might have   decompositions of lower border rank on the boundary.

\section{Koszul  flattenings and lower bounds for Kronecker powers}\label{kossect}

In this section we review Koszul flattenings and prove a result on propagation of Koszul flattening lower bounds under Kronecker products. We will use Koszul flattenings to prove $\uR(T_{skewcw,q}) \geq q+3$ in Proposition \ref{skewbad}. Moreover, we prove Theorem \ref{thm: lower bounds cw powers}: the proof will follow from Theorem \ref{cwbadnews}, Theorem \ref{cwbadnews3} and Corollary \ref{corol: propagation cw}.

Fix bases $\{a_i\}$, $\{b_j\}$, $\{ c_k\}$ of the vector spaces $A,B,C$, respectively; fix an integer $p$. Given a tensor $T = \sum_{ijk} T^{ijk} a_i \otimes b_j \otimes c_k \in A \otimes B \otimes C$, the $p$-th \emph{Koszul flattening} of $T$ on the space $A$ is the linear map
\begin{align*}
T_A^{\ww p}: \Lambda^p A \ot B^* &\to \Lambda^{p+1}A\ot C \\
X\ot \beta & \mapsto \textsum_{ijk}T^{ijk}\beta(b_j)(a_i\wedge X) \ot c_k.
\end{align*}
Then \cite[Proposition 4.1.1]{LanOtt:EqnsSecantVarsVeroneseandOthers} states 
\begin{equation}\label{kozinq}
\ur(T)\geq \frac{\trank(T_{A}^{\ww p})}{\binom {\tdim(A)-1}p}.
\end{equation}

This type of lower bound has a long history. More generally, one considers an embedding of the space $A\ot B\ot C$ into a large space of matrices. Then if a rank-one tensor maps to a rank $q$ matrix, a rank $r$ tensor maps to a rank at most $rq$ matrix, so the size $rq+1$ minors give equations testing for border rank $r$. In this case the size of the matrices is $\binom{\aaa}p\bbb \times \binom{\aaa}{p+1}\ccc$ and a rank-one tensor maps to a matrix of rank $\binom{\aaa-1}p$. Here $\aaa=\tdim A$, $\bbb=\tdim B$ and $\ccc=\tdim C$.

In practice, one considers a subspace ${A'}^*\subseteq A^*$ of dimension $2p+1$ and restricts $T$ (considered as a trilinear form) to ${A'}^* \times  B^* \times C^*$ to get an optimal bound, so the denominator $\binom{\dim(A)-1}{p}$ is replaced by $\binom{2p}{p}$ in \eqref{kozinq}. Equivalently, one considers a linear map $\phi : A \to A'$ and the corresponding Koszul flattening map gives a lower bound for $\uR(\phi(T))$, which, by linearity, is a lower bound for $\uR(T)$.

The case $p=1$ is a straightening of Strassen's equations \cite{Strassen:RanksOptimalComputation}. There are numerous expositions of Koszul flattenings and their generalizations, see, e.g., \cite[\S 7.3]{MR2865915}, \cite[\S 7.2]{BalBerChrGes:PartiallySymRkW}, \cite{MR3781583}, \cite[\S 2.4]{MR3729273}, or \cite{MR3761737}.

We use Koszul flattenings to give the following lower bound on $\uR(T_{skewcq,q})$:
\begin{proposition}\label{skewbad}
For every even $q \geq 2$, $\uR(T_{skewcw,q})\geq q+3$.
\end{proposition}
\begin{proof}
Write $q = 2u$. Fix a space $A' = \langle e_0,e_1,e_2 \rangle$. Define $\phi: A \to A'$ by 
\begin{align*}
\phi( a_0 ) &= e_0, \\
\phi(a_i) &= e_1 \quad \text{for $i = 1 \vvirg u$} , \\
\phi(a_s) &= e_2 \quad \text{for $s = u+1 \vvirg q$}.
 \end{align*}
 
 As an element of $\Lambda^3 A \subseteq A \otimes A \otimes A$, we have $T_{skewcw,q} = a_0 \wedge \sum_{i=1}^u a_i \wedge a_{u+i}$ as in \eqref{eqn: skewcw as skew}.

We prove that for $T = (\phi \otimes \Id_B \otimes \Id_C) (T_{skewcw,q}) \in A' \otimes B \otimes C$, one obtains $\rank( T_{A'}^{\wedge 1} ) = 2(q+2) +1$. This provides the lower bound $\uR( T ) \geq \left\lceil \frac{2(q+2) +1 }{2} \right\rceil = q+3$.

We record the images via $T_{A'}^{\wedge 1}$ of a basis of $A' \otimes B^*$. 
Fix the range of $i = 1 \vvirg u$:
\begin{align*}
T_{A'}^{\wedge 1}( e_0 \otimes \beta_0 ) &= (e_0 \wedge e_1) \otimes 
\textsum_{i=1}^u c_{u+i} - (e_0 \wedge e_2) \otimes \textsum_{i=1}^u c_{i}, \\
T_{A'}^{\wedge 1}( e_0 \otimes \beta_i ) &= (e_0 \wedge e_2) \otimes c_0 , \\
T_{A'}^{\wedge 1}( e_0 \otimes \beta_{u+i} ) &= (e_0 \wedge e_1) \otimes c_0 , 
\\
T_{A'}^{\wedge 1}( e_1 \otimes \beta_{0} ) &= (e_1 \wedge e_2) \otimes 
\textsum_{i=1}^u c_{u+i} , \\
T_{A'}^{\wedge 1}( e_1 \otimes \beta_{i} ) &= (e_0 \wedge e_1) \otimes c_{u+i} 
+ e_1 \wedge e_2 \otimes c_0 , \\
T_{A'}^{\wedge 1}( e_1 \otimes \beta_{u+i} ) &= e_0 \wedge e_1 \otimes c_i, \\
T_{A'}^{\wedge 1}( e_2 \otimes \beta_{0} ) &= (e_1 \wedge e_2) \otimes 
\textsum_{i=1}^u c_{i} , \\
T_{A'}^{\wedge 1}( e_2 \otimes \beta_{i} ) &= e_0 \wedge e_2 \otimes c_{u+i}, 
\\
T_{A'}^{\wedge 1}( e_2 \otimes \beta_{u+i} ) &= (e_0 \wedge e_2) \otimes c_{i} 
- e_1 \wedge e_2 \otimes c_0 . \\
\end{align*}
Notice that the image of $\sum_{i=1}^u (e_1 \otimes \beta_{i}) - \sum_{i=1}^u (e_2 
\otimes \beta_{u+i}) - e_0 \otimes \beta_0$ is (up to scale) $e_1 \wedge e_2 
\otimes c_0$. 

From the contributions above, we deduce that the image of $T_{A'}^{\wedge 1}$ contains the three subspaces
\begin{align*}
 &\langle e_0\wedge e_1, e_0 \wedge e_2 , e_1 \wedge e_2\rangle \otimes \langle c_0  \rangle ,\\ 
 &\langle e_1 \wedge e_2 \rangle \otimes \langle \textsum_{i=1}^u c_{i} , \textsum_{i=1}^u c_{u+i}  \rangle , \\
 &\langle e_0 \wedge e_1  , e_0 \ww e_2 \rangle \otimes \langle c_1 \vvirg c_q\rangle.
\end{align*}
These subspaces are in direct sum, therefore we conclude 
\[
 \rank( T_{A'}^{\wedge 1}) \geq  3 + 2 + 2q = 2q + 5.
\]
\end{proof}

 \subsection{Propagation of lower bounds under Kronecker products}
 
 In \cite{ChrJenZui:NonMultTensorRank,arXiv:1801.04852}, it was shown that generalized flattening lower bounds are multiplicative under the \emph{unflattened} tensor product. The same result does not hold for Kronecker products. However, we provide a partial multiplicativity result for Koszul flattenings lower bounds.
 
A tensor $T\in A\ot B\ot C$, with $\dim B=\dim C$ is {\it $1_A$-generic} if $T(A^*) \subseteq B\ot C$ contains a full rank element. 

\begin{proposition}\label{kronbrlower} 
Let $T_1 \in A_1 \otimes B_1 \otimes C_1$ with $\dim B_1 = \dim C_1$ be a tensor. Let $A'$ be a quotient of $A_1$ with $\dim A' = 2p+1$ and suppose $T_1$ has a Koszul flattening lower bound for border rank $\uR(T) \geq r $ given by 
${T_1}^{\wedge p}_{A'}$. Let $T_2 \in A_2 \otimes B_2 \otimes C_2$, with $\dim B_2 = \dim C_2 = \bfb_2$ be $1_{A_2}$-generic.   Then 
\begin{equation}\label{kstar}
 \uR(T_1 \boxtimes T_2) \geq \left\lceil  \frac{ \rank ({T_1}^{\wedge p}_{A'}) 
\cdot \bfb_2 }{ \binom{2p}{p}} \right\rceil.
\end{equation}
In particular, if $\frac{ \rank ({T_1}^{\wedge p}_{A'}) }{ \binom{2p}{p}}\in \BZ$, then  $\uR(T_1 \boxtimes T_2)\geq r\bbb_2$. 
\end{proposition}
\begin{proof}
Identify $T_1$ with its image in $A' \otimes B_1 \otimes C_1$. The lower bound for $T_1$ is 
\[
 \uR(T_1) \geq \left\lceil \frac{\rank ({T_1}^{\wedge p}_{A'}) }{ 
\binom{2p}{p}} \right\rceil.
\]
Let $\alpha \in A_2^*$  be such that 
$T(\alpha ) \in B_2 \otimes C_2$ has full rank  $\bfb_2$, which exists by 
$1_{A_2}$-genericity.   Define 
$\psi : A' \otimes A_2 \to A'$   by   $\psi = \Id_{A'} \otimes \alpha$ and set $\Psi:=\psi\ot \Id_{B_1\ot C_1\ot B_2\ot C_2}$. Then 
$(\Psi( T_1 \boxtimes T_2) ^{\wedge p}_{A'} )$ provides the desired lower 
bound.

Indeed,   the linear map $( \Psi( T_1 \boxtimes T_2) ^{\wedge p}_{A'}  )$ 
coincides with ${T_1}^{\wedge p}_{A'} \boxtimes T_1(\alpha)$. Since matrix rank 
is multiplicative under Kronecker product, we   conclude.  
\end{proof}

\subsection{A lower bound for the Kronecker square of $T_{cw,q}$}\label{subsection: first proof of cwsquare}

In this section, we give a proof of the first statement in Theorem \ref{thm: lower bounds cw powers}.

The statement for $q = 2$, can be checked explicitly. The lower bound $\uR(T_{cw,2}^{\boxtimes 2}) \geq 15$ follows from the $p=2$ Koszul flattening lower bound and coincides with the current best known lower bound for the border rank of the $3 \times 3$ permanent polynomial. The upper bound is immediate by submultiplicativity.

\begin{theorem}\label{cwbadnews}
 Let $q \geq 3$. Then $\uR(T_{cw,q}^{\boxtimes 2}) = (q+2)^2$.
\end{theorem}
\begin{proof}
Recall the expression of $T_{cw,q}$ from \eqref{littlecw}.
 When  $q = 3$, the result is true by a direct calculation using the $p=2$ Koszul flattening with a sufficiently generic restriction $A \to \bbC^5$. 
 
 Assume $q>3$. Write $a_{ij} = a_i \otimes a_j \in A^{\otimes 2}$ and similarly for $B^{\otimes 2}$ and $C^{\otimes 2}$. Let $A' = \langle e_0,e_1,e_2 \rangle$ and define the linear map $\phi_2: A^{\otimes 2} \to A'$ by
\begin{align*}
\phi_2( a_{00} ) &= \phi_2(a_{01}) = \phi_2(a_{10})= e_0 + e_1, \\
\phi_2(a_{11}) &= e_0  , \\
\phi_2(a_{02}) &= \phi_2(a_{20}) = e_1 + e_2 \\
\phi_2(a_{33}) &= \phi_2(a_{21}) = e_2 \\
\phi_2(a_{0i}) &= \phi_2(a_{i0}) = e_1 \quad \text{for $i = 3 \vvirg q$}\\
\phi_2(a_{ij}) &= 0 \quad \text{for all other pairs $(i,j)$}.
 \end{align*}

Write $T_q:= \phi_2( T_{cw,q}^{\boxtimes 2}) \in A' \otimes B^{\otimes 2} \otimes C^{\otimes 2}$.  Consider the  $p=1$ Koszul flattening $(T_q)^{\wedge 1}_{A'}  : A' \otimes {B^{\otimes 2}}^*   \to \Lambda^2 A' \otimes C^{\otimes 2}$.

We are going to prove that $\rank( (T_q)^{\wedge 1}_{A'}  ) = 
2(q+2)^2$. This provides the lower bound $\uR(T_{cw,q}^{\boxtimes 2}) \geq 
(q+2)^2$ and equality follows because the upper bound is immediate by submultiplicativity.

We proceed by induction on $q$.  When $q=4$
one does a direct computation with the $p=1$ Koszul flattening, which is left to the reader, and 
which provides the base of the induction.  

Write $W_j = a_0 \otimes b_j \otimes c_j + a_j \otimes b_0 \otimes c_j + a_j \otimes b_j \otimes c_0$. Then $T_{cw,q} = \sum_{j=1}^q W_j$, so that 
$T_{cw,q}^{\boxtimes 2} = \sum_{ij} W_i \boxtimes W_j$.

If $q \geq 4$, write $T_{cw,q} = T_{cw, q-1} + W_q$, so 
$T_{cw,q}^{\boxtimes 2} = T_{cw,q-1}^{\boxtimes 2} + T_{cw,q-1} \boxtimes W_q + 
W_q \boxtimes T_{cw,q-1} + W_q \boxtimes W_q$. Let $S_q = \phi_2 ((T_{cw,q-1}\boxtimes  W_q + W_q \boxtimes 
T_{cw,q-1} + W_q \boxtimes W_q) )$.

Write $U_{1} =  A' \otimes \langle \beta_{ij} : i,j = 0 \vvirg q-1 \rangle$ and 
$U_2 =  A' \otimes \langle \beta_{qi}, \beta_{iq} : i=0 \vvirg q \rangle$ so 
that $U_1 \oplus U_2 = A' \otimes B^{\ot 2 *}$. Similarly, define $V_1 = \Lambda^2 A' 
\otimes \langle c_{ij} : i,j = 0 \vvirg q-1 \rangle$ and $V_2 = \Lambda^2 A' 
\otimes \langle c_{qi}, c_{iq} : i=0 \vvirg q \rangle$, so that $V_1 \oplus 
V_2 = \Lambda^2 A' \otimes C^{\ot 2} $. Observe that $(T_{q-1} )_{A'}^{\wedge 1}$ is identically $0$ on $U_2$ and its image is contained in 
$V_1$. Moreover, the image of  $U_1$ under  $(S_q)^{\wedge 1}_{A'}$  is contained in 
$V_1$. Representing the Koszul flattening in blocks, we have
\[
( T_{ q-1} )_{A'}^{\wedge 1} = \left[ \begin{array}{cc} M_{11} & 
0 \\ 0 & 0 \end{array} \right] \qquad
({S_q})^{\wedge 1}_{A'} = \left[ \begin{array}{cc} N_{11} & N_{12} \\ 0 & N_{22} 
\end{array} \right]
\]
therefore $\rank(  (T_{q})^{\ww 1}_{A'} ) \geq  \rank ( M_{11} + N_{11}) + 
\rank(N_{22})$.

First, we prove that $\rank ( M_{11} + N_{11}) \geq \rank(M_{11}) = 2(q+1)^2$. 
This follows by a degeneration argument.

Consider the degeneration given by the linear maps $(g_\eps,h_\eps) \in GL(B^{\otimes 2}) \times GL(C^{\otimes 2})$ with 
\[
\begin{array}{rlcrl}
 g_\eps : b_{iq} & \mapsto \eps b_{iq} & \qquad & h_\eps : c_{iq} & \mapsto \eps c_{iq} \\
    b_{qi} &\mapsto \eps b_{qi} & & c_{qi} & \mapsto \eps c_{qi}\\  
 b_{ij} &\mapsto b_{ij} \quad \text{if $i,j \neq q $} & & c_{ij} & \mapsto c_{ij} \quad \text{if $i,j \neq q $} \\
 \end{array}.
\]
Let $T_{q,\eps} = (g_\eps ,h_\eps ) \cdot T_{q}$. We have $T_{q,\eps} = T_{q-1} + S_{q,\eps}$ where $S_{q,\eps} = (g_\eps ,h_\eps ) \cdot S_{q}$. In particular $\lim_{\eps \to 0} S_{q,\eps} = 0$. Moreover, the degeneration preserves the spaces $U_1,U_2,V_1,V_2$, therefore the Koszul flattening of $T_{q,\eps}$ has the same block structure as the one of $T_{q}$ with
\[
( S_{q,\eps})^{\wedge 1}_{A'} = \left[\begin{array}{cc}N_{11}(\eps) & N_{12}(\eps) \\ 0 & N_{22}(\eps) \end{array}\right].
\]
Since $\lim_{\eps \to 0} S_{q,\eps} = 0$, we have $ \lim_{\eps \to 0} N_{11}(\eps) \to 0$. The value of $\rank(M_{11} + N_{11}(\eps))$ is constant for (generic) $\eps \neq 0$, and by semicontinuity we obtain
\[
\rank(M_{11}) = \rank(\lim_{\eps \to 0}( M_{11} + N_{11}(\eps)) \leq \rank( M_{11} + N_{11}).
\]
By the induction hypothesis $\rank(M_{11}) = 2(q+1)^2$, threfore $\rank (M_{11} + N_{11}) \geq 2(q+1)^2$.

We show that $\rank(N_{22}) = 2(2q+3)$. The following equalities are modulo 
$V_1$. Moreover, each equality is modulo the tensors resulting from the previous 
ones. They are all straightforward applications of 
the Koszul flattening map, which in these cases, can always be performed on some 
copy of $W_i \boxtimes W_j$.
\begin{align*}
 ({S_q})^{\wedge 1}_{A'} ( e_1 \otimes \beta_{qj}) &\equiv e_1 \wedge e_0 
\otimes c_{qj}  \quad \text{ for $j = 3 \vvirg q$} \\
 ({S_q})^{\wedge 1}_{A'} ( e_1 \otimes \beta_{jq}) &\equiv e_1 \wedge e_0 
\otimes c_{jq}  \quad \text{ for $j = 3 \vvirg q$} \\
 ({S_q})^{\wedge 1}_{A'} ( e_0 \otimes \beta_{3q}) & \equiv e_0 \wedge e_1 
\otimes c_{0q} \\
 ({S_q})^{\wedge 1}_{A'} ( e_0 \otimes \beta_{q3}) & \equiv e_0 \wedge e_1 
\otimes c_{q0} \\
 ({S_q})^{\wedge 1}_{A'} ( e_0 \otimes \beta_{q1}) & \equiv e_0 \wedge e_1 
\otimes c_{q1} \\
 ({S_q})^{\wedge 1}_{A'} ( e_0 \otimes \beta_{1q}) & \equiv e_0 \wedge e_1 .
\otimes c_{1q} \\
\end{align*}
Further passing modulo $\langle e_0 \wedge e_1 \rangle \otimes C$, we obtain
\begin{align*}
 ({S_q})^{\wedge 1}_{A'} ( e_0 \otimes \beta_{0q}) & \equiv e_0 \wedge e_2 
\otimes c_{2q} \\
 ({S_q})^{\wedge 1}_{A'} ( e_0 \otimes \beta_{q0}) & \equiv e_0 \wedge e_2 
\otimes c_{q2} \\
 ({S_q})^{\wedge 1}_{A'} ( e_0 \otimes \beta_{q2}) & \equiv e_0 \wedge e_2 
\otimes c_{0q} \\
 ({S_q})^{\wedge 1}_{A'} ( e_0 \otimes \beta_{2q}) & \equiv e_0 \wedge e_2 
\otimes c_{q0} \\
 ({S_q})^{\wedge 1}_{A'} ( e_1 \otimes \beta_{20}) & \equiv e_1 \wedge e_2 
\otimes c_{0q} \\
 ({S_q})^{\wedge 1}_{A'} ( e_1 \otimes \beta_{02}) & \equiv e_1 \wedge e_2 
\otimes c_{q0} \\
 ({S_q})^{\wedge 1}_{A'} ( e_1 \otimes \beta_{q0}) & \equiv e_1 \wedge e_2 
\otimes c_{2q} \\
 ({S_q})^{\wedge 1}_{A'} ( e_1 \otimes \beta_{0q}) & \equiv e_1 \wedge e_2 
\otimes c_{q2}, \\
\end{align*}
and modulo the above, 
\begin{align*}
 ({S_q})^{\wedge 1}_{A'} ( e_2 \otimes \beta_{qj}) &\equiv e_2 \wedge (e_0 + 
e_1) \otimes c_{qj}  \quad \text{ for $j = 3 \vvirg q$} \\
 ({S_q})^{\wedge 1}_{A'} ( e_2 \otimes \beta_{jq}) &\equiv e_2 \wedge (e_0 + 
e_1) \otimes c_{jq}  \quad \text{ for $j = 3 \vvirg q$} \\
 ({S_q})^{\wedge 1}_{A'} ( e_2 \otimes \beta_{q1}) &\equiv e_2 \wedge (e_0 + 
e_1) \otimes c_{q1} \\
 ({S_q})^{\wedge 1}_{A'} ( e_2 \otimes \beta_{1q}) &\equiv e_2 \wedge (e_0 + 
e_1) \otimes c_{1q} .\\
\end{align*}
Finally passing modulo $\langle e_1 \wedge e_2 \rangle$, we have
\begin{align*}
 ({S_q})^{\wedge 1}_{A'} ( e_2 \otimes \beta_{q0}) &\equiv e_2 \wedge e_0 
\otimes c_{q1}  \\
 ({S_q})^{\wedge 1}_{A'} ( e_2 \otimes \beta_{0q}) &\equiv e_2 \wedge e_0 
\otimes c_{1q} . \\
\end{align*}
All the tensors listed above are linearly independent. Adding all the 
contributions together, we obtain
\[
 \rank ( ({S_q})^{\wedge 1}_{A'} ) = [ 2(q-3) + 1 ] + 4 + 8 + 2 + [2(q-3)+1] + 4 
= 2(2q+3)
\]
as desired, and since $2(q+3)^2=2(q+1)^2+2(2q+3)$,  this concludes the proof. 
 \end{proof}
 
 We will provide a second proof of Theorem \ref{cwbadnews}, which will generalize to the proof of Theorem \ref{cwbadnews3}. More precisely, we will give a representation-theoretic argument to compute the rank of the Koszul flattening map considered in the proof above. The same representation-theoretic technique will apply for the third Kronecker power.

\subsection{A short detour on computing ranks of equivariant maps}\label{subsec: schur recap}

We briefly explain how to exploit Schur's Lemma (see, e.g., \cite[\S1.2]{FH}) 
to compute the rank of an equivariant linear map. This is a standard technique, 
used extensively e.g.,  in \cite{MR3376667,GesIkPa:GCTMatrixPowering} and will  
 reduce the proof of Theorems \ref{cwbadnews} and \ref{cwbadnews3} to the 
computation of the ranks of specific linear maps in small dimension.
 
Let $G$ be a reductive group. In the proof of Theorems \ref{cwbadnews} and 
\ref{cwbadnews3}, $G$ will be the product of symmetric groups. Let $\Lambda_G$ 
be the set of irreducible representations of $G$. For $\lambda \in \Lambda_G$,  
 let  $W_\lambda$ denote  the corresponding irreducible module.

Suppose $U,V$ are two representations of $G$. Write $U= \bigoplus_{\lambda \in 
\Lambda_G} W_\lambda^{\oplus m_\lambda}$, $V = \bigoplus_{\lambda \in 
\Lambda_G} W_\lambda^{\oplus \ell_\lambda}$, where $m_\lambda$ is the 
multiplicity of   $W_\lambda$ in $U$ and   $\ell_\lambda$ is the multiplicity of 
$W_{\lambda}$ in $V$. The direct summand corresponding to $\lambda$ is called 
the \emph{isotypic component} of type $\lambda$.

Let $f : U \to V$ be a $G$-equivariant map. By Schur's Lemma \cite[\S1.2]{FH}, 
$f$ decomposes   as $f = \oplus f_\lambda$, where 
$f_\lambda : W_\lambda^{\oplus m_\lambda} \to W_\lambda^{\oplus \ell_\lambda}$ are $G$-equivariant. 
Consider multiplicity spaces $M_\lambda, L_\lambda$ with $\dim M_\lambda = 
m_\lambda$ and $\dim L_\lambda = \ell_\lambda$ so that   $W_\lambda^{\oplus 
m_\lambda} \simeq M_\lambda \otimes W_\lambda$ as a $G$-module, where $G$ acts 
trivially on $M_\lambda$ and similarly 
$W_\lambda^{\oplus \ell_\lambda} \simeq L_\lambda \otimes W_\lambda$.

By Schur's Lemma, the map   $f_\lambda : M_\lambda \otimes W_\lambda \to 
L_\lambda \otimes W_\lambda$ decomposes as $f_\lambda = \phi_\lambda \otimes 
\Id_{[\lambda]}$, where $\phi_\lambda : M_\lambda \to L_\lambda$. Thus  
$\rank(f)$ can be expressed in terms of $\rank(\phi_\lambda)$ and the dimension of the multiplicity spaces $W_\lambda$ for $\lambda \in 
\Lambda_G$:
\[
 \rank(f) = \textsum_\lambda \rank(\phi_\lambda) \cdot \dim W_\lambda.
\]

The ranks $ \rank(\phi_\lambda)$ can be computed via restrictions of 
$f$. For every $\lambda$, fix a nonzero vector $w_\lambda \in W_\lambda$, so that 
$M_\lambda \otimes \langle w_\lambda\rangle$ is a subspace of $U$. Here and in 
what follows, for a subset $X\subset V$,
$\langle X\rangle$ denotes the span of $X$. Then  the rank of the restriction 
of $f$ to $M_\lambda \otimes \langle w_\lambda \rangle$ coincides with the rank 
of $\phi_\lambda$. 

The second proof  of Theorem  \ref{cwbadnews} and proof of Theorem  \ref{cwbadnews3} will follow the algorithm
described above, exploiting the symmetries of $T_{cw,q}$. Consider the action 
of the symmetry group $\frakS_{q}$ on $A \otimes B \otimes C$ defined by 
permuting the basis elements with indices $\{1 \vvirg q\}$. More precisely, a 
permutation $\sigma \in 
\frakS_q$  induces the linear map defined by $\sigma (a_i) = a_{\sigma(i)}$ for 
$i= 1\vvirg q$ and $\sigma(a_0) = a_0$. The group $\frakS_q$ acts on $B,C$ 
similarly, and the simultaneous action on the three factors defines {an 
$\frakS_q$-action on $A \otimes B \otimes C$}. The tensor $T_{cw,q}$ is 
invariant under this action.

\subsection{Second Proof of Theorem \ref{cwbadnews}}\label{secondpf}

We use the method explained in Section \ref{subsec: schur recap} to give a representation-theoretic proof of Theorem \ref{cwbadnews}.

\begin{proof}[Proof of Theorem \ref{cwbadnews}]
As before, the case $q=3$ can be verified explicitly. For $q \geq 4$, we apply the $p=1$  Koszul flattening map  to the same restriction of  $T_{cw,q}^{\boxtimes 2}$ as the first proof, although to be consistent with the code on the website, we use the less appealing swap of the roles of $a_2$ and $a_3$ in the projection $\phi_2$.
  
The tensor $T_{cw,q}$ is invariant under the action of $\frakS_q$ acting on the indices $\{1 \vvirg q\}$ of the basis elements of $\bbC^{q+1}$. Therefore  
$T_{cw,q}^{\boxtimes 2}$ is invariant under the action of   $\frakS_{q} 
\times \frakS_{q}$ on $A^{\otimes 2} \otimes B^{\otimes 2} \otimes C^{\otimes 2}$. Let 
$\Gamma := \frakS_{q-3} \times \frakS_{q-3}$ 
where $\frakS_{q-3}$ is the permutation group on $\{ 4 \vvirg q\}$; $T_{cw,q}^{\boxtimes 2}$ is invariant under the action of $\Gamma$. 

Moreover, the projection $\phi_2$ is invariant under the action of $\Gamma$. 

In general, the map $A \otimes B \otimes C \to \Hom( B^* \otimes \Lambda^p A , C \otimes \Lambda^{p+1}A)$ is equivariant for the action of $GL(A) \times GL(B) \times GL(C)$. Using this fact, and the invariancy with respect to $\Gamma$ described above, we deduce $(\phi_2(T_{cw,q}^{\boxtimes 2}))_{A'}^{\ww 1}$ is $\Gamma$-equivariant.

We now apply the method described in \S\ref{subsec: schur recap} to compute $\trank((T_q)_{A'}^{\ww 1})$. 

Let $[\triv]$ denote the trivial $\frakS_{q-3}$-representation and  let  $V$ denote  the standard representation, that is the Specht module associated to the partition $(q-4,1)$ of $q-3$. We have $\dim [\triv] = 1$ and $\dim V = q-4$. When $q=4$ only the trivial representation appears.

The spaces $B,C$ are isomorphic as $\frakS_{q-3}$-modules and they decompose as $B = C = [\triv]^{\oplus 5} \oplus V$. After fixing a $5$-dimensional multiplicity space $\bbC^5$ for the trivial isotypic component, we write ${B^*} = C = \bbC^5 \otimes [\triv] \oplus V$. To distinguish the two $\FS_{q-3}$-actions, we write ${B^*}^{\otimes 2} = ([\triv]_L^{\oplus 5} \oplus V_L)\ot ([\triv]_R^{\oplus 5} \oplus V_R)$ and similarly for $C^{\otimes 2}$

Thus,  
\begin{align*}
 {B^*}^{\otimes 2} = C^{\otimes 2} = & {\bbC^5}^{\otimes 2} \otimes ([\triv]_L \otimes [\triv]_R) \oplus \\
 & \bbC^5 \otimes ( [\triv]_L \otimes V_R) \quad \oplus \\
 & \bbC^5  \otimes (V_L \otimes [\triv]_R)  \quad \oplus \\
 &(V_L \otimes V_R).
\end{align*}

Write $W_1 \vvirg W_4$ for the four irreducible representations in the 
decomposition above and let $M_1 \vvirg M_4$ be the four corresponding 
multiplicity spaces. 

Recall from \cite{Fulton:YoungTableaux} that a basis of $V$ is given by 
standard 
Young tableaux of shape $(q-4,1)$ ({with} entries in $4 \vvirg q$ for 
consistency with the action of $\frakS_{q-3}$); let $w_{std}$ be the vector 
corresponding to the standard tableau having $4,6 \vvirg q$ in the first row 
and $5$ in the second row. We refer to \cite[\S7]{Fulton:YoungTableaux} for the 
straightening laws of the tableaux. Let $w_{\triv}$ be a generator of the 
trivial representation $[\triv]$. Writing $\bbC^{q+1} = \langle e_0 \vvirg e_q\rangle$, we explicitly have $w_{std} = e_5-e_4$ and the multiplicity space $5$-dimensional multiplicity space of the trivial representation is $\langle e_0 \vvirg e_3, \textsum_{4}^q e_j \rangle$. 

For each of the four isotypic components in the decomposition above,  we fix a vector $w_i \in W_i$ and explicitly realize  the subspaces $M_i \otimes \langle w_i \rangle$ of ${B^*}^{\otimes 2}$ as follows:
\[
 \begin{array}{cccc}
  W_i & w_i &  \dim M_i & M_i \otimes \langle w_i \rangle \\ \midrule  \\ 
  {[\triv]_L \otimes [\triv]_R } & w_{\triv} \otimes w_{\triv} & 25 & 
\begin{smallmatrix}
                                                                   \langle 
\beta_{ij} : i,j =0 \vvirg 3 \rangle \oplus \\
                                                                   \langle 
\sum_{j=4}^q \beta_{ij} : i = 0 \vvirg 3\rangle  \oplus \\
                                                                   \langle 
\sum_{i=4}^q \beta_{ij} : j = 0 \vvirg 3\rangle \oplus \\
                                                                   \langle 
\sum_{i,j=4}^q \beta_{ij} \rangle
                                                                  
\end{smallmatrix} \\
                                                                  \\
{[\triv]_L \otimes V_R } &  w_\triv \otimes w_{std} & 5 & \begin{smallmatrix}
                                                       \langle \beta_{i5} - 
\beta_{i4} : i = 0 \vvirg 3 \rangle \oplus \\
                                                       \langle \sum_{i=4}^q 
(\beta_{i5}-\beta_{i4}) \rangle
                                                      \end{smallmatrix} \\
\\
{V_L \otimes [\triv]_R  } &  w_{std} \otimes w_\triv & 5 & \begin{smallmatrix}
                                                       \langle \beta_{5j} - 
\beta_{4j} : j = 0 \vvirg 3 \rangle \oplus \\
                                                       \langle \sum_{j=4}^q 
(\beta_{5j}-\beta_{4j}) \rangle
                                                      \end{smallmatrix} \\
                                                      \\
V_L \otimes V_R & w_{std} \otimes w_{std} & 1 & \langle \beta_{55} - \beta_{45} 
- 
\beta_{54} + \beta_{44} \rangle.
 \end{array} 
\]

The subspaces in $C^{\otimes 2}$ are realized similarly. 

Since $(T_{cw,q}^{\boxtimes 2})^{\wedge 1}_{A'}$ is $\Gamma$-equivariant, by Schur's Lemma, it has the isotypic decomposition $(T_{cw,q}^{\boxtimes 2})^{\wedge 1}_{A'} = f_1 \oplus f_2 \oplus f_3 \oplus f_4$,  where 
\begin{equation}\label{eq: diagonal blocks square}
f_i : A' \otimes (M_i \otimes W_i) \to \Lambda^2 A' \otimes (M_i \otimes W_i).
\end{equation}
As explained in \S\ref{subsec: schur recap}, it suffices to compute the ranks of the four restrictions $\Phi_i : A' \otimes M_i \otimes \langle w_i\rangle  \to 
\Lambda^2 A'\otimes M_i \otimes \langle w_i \rangle$ to the multiplicities spaces.

The four matrices representing $\Phi_1 \vvirg \Phi_4$ are computed by a routine which exploits their structure. The script to compute the matrices and their ranks is available in Appendix D of the Supplementary Material. The method to compute the matrices is explained in Section \ref{section: boxparametrized}.

The script provides an expression for the entries of the matrices $\Phi_i$ which are univariate polynomials in $q$ up to a global univariate polynomial factor. The expressions are valid for $q \geq 5$. The rank of the Koszul flattening in the cases $q= 3$ and $q = 4$ is computed directly.

 We determine a lower bound on $\rank(\Phi_i)$ by computing a matrix $P_i \cdot \Phi_i \cdot Q_i$, where $P_i$ is a rectangular matrix whose entries are rational functions of $q$ (well defined for $q \geq 5$) and $Q_i$ is a rectangular matrix whose entries are constant. The resulting matrix $P_i \cdot \Phi_i \cdot Q_i$ is a square matrix, upper triangular with $\pm 1$ on the diagonal, so that the size of $P_i \Phi_i Q_i$ gives a lower bound on $\rank( \Phi_i)$.

We summarize the results of the script in the following table.
\[
\def\arraystretch{1.6}
 \begin{array}{ccccc}
\text{$W_i$} & \dim W_i & \dim M_i & \rank ( \Phi_i ) & \text{contribution to total rank}  \\ \midrule
{[\triv]_L \otimes [\triv]_R } & 1 & 25 & 72 & 72 \\
{[\triv]_L \otimes V_R } & q-4 & 5 & 12 & 12(q-4)\\
{V_L \otimes [\triv]_R } & q-4 & 5 & 12 & 12(q-4)\\
{V_L \otimes V_R } & (q-4)^2& 1 &  2 & 2(q-4)^2 
\end{array}
\]

Adding the total contributions, we obtain 
\begin{align*}
 \rank (T^{\wedge 1}_{A'}) & = 2 \cdot (q-4)^2 + 12 \cdot (q-4) + 12 \cdot (q-4) + 72 \cdot 1= 2(q+2)^2.
\end{align*}
This concludes the proof of Theorem \ref{cwbadnews}. 
\end{proof}

\subsection{A lower bound for the Kronecker cube of $T_{cw,q}$}\label{subsec: proof kron cube Tcw}

In this section, we use the method explained in Section \ref{subsec: schur recap} and illustrated in Section \ref{secondpf} to prove the second part of Theorem \ref{thm: lower bounds cw powers}.

\begin{theorem}\label{cwbadnews3}
 Let $q \geq 5$. Then $\uR(T_{cw,q}^{\boxtimes 3}) = (q+2)^3$.
\end{theorem}
\begin{proof}
 We will give a lower bound on $\uR(T_{cw,q}^{\boxtimes 3})$ by computing its Koszul flattening for $p =2$. Write $a_{ijk} = a_i \otimes a_j \otimes a_k \in A^{\otimes 3}$ and similarly for $B^{\otimes 3}$ and $C^{\otimes 3}$. Let $\{ \alpha_{ijk} \} \subseteq A^{* \otimes 3}$ be the dual basis to $\{ a_{ijk} \} \subseteq A^{\otimes 3}$. Let $A' = \langle e_0 \vvirg e_4\rangle$ be a $5$-dimensional space and let $\{ e^0 \vvirg e^4\}$ be the dual basis of $\{e_0 \vvirg e_4\}$ and define $\phi_3 : A^{\otimes 3} \to A'$ to be the linear map whose transpose $\phi_3^T : {A'}^* \to A^{*\otimes 3}$ is given by 
\begin{align*}
 \phi_3^T(e^0) &= \alpha_{000} \\
 \phi_3^T(e^1) &= \textsum_{i=1}^q (\alpha_{i00} +\alpha_{0i0} +\alpha_{00i}) \\
 \phi_3^T(e^2) &= \alpha_{001} +\alpha_{010} +\alpha_{012} +\alpha_{102} +\alpha_{110}+\alpha_{121} +\alpha_{200} +\alpha_{211} \\
 \phi_3^T(e^3) &=  \alpha_{022 } +\alpha_{030 } +\alpha_{031 } +\alpha_{100 } +\alpha_{103 } - \alpha_{120 } +\alpha_{210 } +\alpha_{212 } +\alpha_{300} \\
 \phi_3^T(e^4) &= \alpha_{002 } +\alpha_{004 } +\alpha_{011 } +\alpha_{014 } +\alpha_{020 }  +\alpha_{023 } +\alpha_{032 } +\alpha_{040 } +\alpha_{100 } +\alpha_{122 }  +\alpha_{220 } +\alpha_{303 } .\\
\end{align*}
Let $T_q = \phi_3 ( T_{cw,q}^{\boxtimes 3}) \in A' \otimes B^{\otimes 3} \otimes C^{\otimes 3}$ and consider the Koszul flattening 
\[
 (T_q)^{\wedge 2}_{A'} : \Lambda^2 A' \otimes {B^*}^{\otimes 3} \to \Lambda^3 A' \otimes C^{\otimes 3}.
\]

We will show $\rank ( (T_q)^{\wedge 2}_{A'}) = 6 (q+2)^3$, which implies $\uR(T_{cw,q}^{\boxtimes 3}) \geq (q+2)^3$.

We employ the same method as in Section \ref{secondpf} in the case of $T_{cw,q}^{\boxtimes 2}$. The Koszul flattening is equivariant for the action of $\Gamma = \frakS_{q-4}^{\times 3}$ where $\frakS_{q-4}$ acts on $\{ 5 \vvirg q\}$. In particular $\bbC^{q+1}$ splits under the action of $\frakS_{q-4}$ into a $6$-dimensional subspace of invariants $\bbC^6 \otimes [\triv] = \langle e_0 \vvirg e_4, e_5 + \cdots + e_q \rangle$ and a copy of the standard representation $V = \langle e_i - e_5 : i = 6\vvirg q\rangle$, with $\dim V = q-5$.

Hence, the spaces $B^{\otimes 3}$ and $C^{\otimes 3}$ split into the direct sum of $8$ isotypic components for the action of $\Gamma$ as follows (we use indices $1,2,3$ to denote the trivial or the standard representation on the first, second or third factor):
\begin{align*}
 {B^*}^{\otimes 3} \simeq C^{\otimes 3} = & (\bbC^{6})^{\otimes 3} \otimes ([\triv]_1 \otimes [\triv]_2 \otimes [\triv]_3) \oplus  \\
 & (\bbC^{6})^{\otimes 2} \otimes \Bigl[  ([\triv]_1 \otimes [\triv]_2 \otimes V_3) \oplus \\
 & \phantom{(\bbC^{6})^{\otimes 2} \otimes \Bigl[} ([\triv]_1 \otimes V_2 \otimes [\triv]_3) \oplus \\
 & \phantom{(\bbC^{6})^{\otimes 2} \otimes \Bigl[} (V_1 \otimes [\triv]_2 \otimes [\triv]_3) \Bigr] \oplus \\
 & (\bbC^{6}) \otimes \Bigl[([\triv]_1 \otimes V_2 \otimes V_3) \oplus  \\
& \phantom{(\bbC^{6}) \otimes \Big[}(V_1 \otimes V_2 \otimes [\triv]_3) \oplus\\ 
& \phantom{(\bbC^{6}) \otimes \Big[} (V_1 \otimes [\triv]_2 \otimes V_3) \Bigr] \oplus \\
& V_1 \otimes V_2 \otimes V_3
\end{align*}

Similarly to the square case, for each of the eight isotypic components, we consider $w_i \in W_i$ where $W_i$ is the corresponding irreducible and we compute the rank of the restriction $\Psi_i : \Lambda^2 A' \otimes M_i \otimes \langle w_i \rangle \to \Lambda^3 A' \otimes M_i \otimes \langle w_i \rangle$ of the Koszul flattening.

The matrices representing the maps $\Psi_i$ are computed exploiting the structure of the tensors involved, following the method described in Section \ref{section: boxparametrized}. The expression computed by the script is valid for $q \geq 6$. The case $q = 5$ is computed explicitly. Their ranks are computed by reducing $\Psi_i$ to a triangular matrix as in the previous case.

The ranks of the restrictions are recorded in the following table:
\[
\def\arraystretch{1.6}
 \begin{array}{ccccc}
   W_i & \dim W_i &  \dim M_i & \rank( \Psi_i) & \text{total contribution}\\ \midrule
  {[\triv]_1 \otimes [\triv]_2 \otimes [\triv]_3} & 1 & 6^3 = 216 & 2058 & 2058 \\
  \\
  \begin{array}{c}
  [\triv]_1 \otimes [\triv]_2 \otimes V_3 \\
  \text{(and permutations)}
  \end{array}   & \begin{array}{c} (q-5) \\ ( \text{three times} ) \end{array} &\begin{array}{c}  6^2  = 36  \\ ( \text{three times} ) \end{array} &  \begin{array}{c} 294 \\ ( \text{three times} ) \end{array} & 3\cdot 294(q-5) \\
  \\
  \begin{array}{c}
  [\triv]_1 \otimes V_2 \otimes V_3 \\
  \text{(and permutations)}
  \end{array}   & \begin{array}{c} (q-5)^2 \\ ( \text{three times} ) \end{array} &\begin{array}{c}  6  \\ ( \text{three times} ) \end{array} & \begin{array}{c} 42 \\ ( \text{three times} ) \end{array} & 3 \cdot 42(q-5)^2 \\
  \\
V_1 \otimes V_2 \otimes V_3 & (q-5)^3 & 1 & 6 & 6(q-5)^3\\
 \end{array} 
\]

Adding all the contributions together, we obtain
\begin{align*}
 \rank (T^{\wedge 2}_{A'}) = &6(q-5)^3 + 3\cdot 42(q-5)^2 +  3\cdot 294(q-5)+ 2058 \cdot 1 = 6 \cdot (q+2)^3.
\end{align*}
This concludes the proof of Theorem \ref{cwbadnews3}.
 \end{proof}

 The third part of Theorem \ref{thm: lower bounds cw powers} is a consequence of Proposition \ref{kronbrlower} and Theorem \ref{cwbadnews3} for the case $q \geq 5$ and  Proposition \ref{kronbrlower} and Theorem \ref{cwbadnews} in the case $q = 4$. We record it explicitly in the following Corollary
 
\begin{corollary}\label{corol: propagation cw}
  For all $q>4$  and all $N$, $\ur(T_{cw,q}^{\boxtimes N})\geq (q+1)^{N-3}(q+2)^3$, and $\ur(T_{cw,4}^{\boxtimes N})\geq 36\times 5^{N-2}$.
 \end{corollary}
 \begin{proof}
If $q > 4$, let $T_1 = T_{cw,q}^{\boxtimes 3}$ and $T_2 = T_{cw,q}^{\boxtimes N-3}$. Since $T_2$ is $1_A$-generic, the lower bound $\ur(T_{cw,q}^{\boxtimes N})\geq (q+1)^{N-3}(q+2)^3$ follows by Proposition \ref{kronbrlower}.

If $q = 4$, let $T_1 = T_{cw,q}^{\boxtimes 2}$ and $T_2 = T_{cw,q}^{\boxtimes N-2}$. Again, since $T_2$ is $1_A$-generic, the lower bound $\ur(T_{cw,4}^{\boxtimes N})\geq (4+2)^2\times 5^{N-2} =  36\times 5^{N-2}$ follows by Proposition \ref{kronbrlower}.
 \end{proof}

 This concludes the proof of Theorem \ref{thm: lower bounds cw powers}.

\section{Upper bounds for Waring rank and border Waring rank of $\det_3$}\label{section: determinant upper bounds}
 
In this section, we give a proof of Theorem \ref{thm: waring rank and border rank det3}. 

We briefly recall the definition of Waring rank and border Waring rank. A symmetric tensor $T\in S^d \bbC^m \subseteq {\BC^m}^{\otimes d}$ has {\it Waring rank one} if $T=a^{\otimes d}$ for some $a\in \bbC^m$. The {\it Waring rank} of $T$, denoted $\bfR_S(T)$, is the smallest $r$ such that $T$ is sum of $r$ tensors of  Waring rank one. The {\it border Waring rank} of $T$, denoted $\ur_S(T)$, is the smallest $r$ such that $T$ is limit of a sequence of tensors of Waring rank $r$. If $T$ is regarded as a homogeneous polynomial of degree $d$, then $a \in \bbC^m$ can be regarded as a linear form and $a^{\otimes d}$ coincides with the $d$-th power of $a$: in this setting, the Waring rank is the minimum number of summands in an expression of $T$ as sum of powers of linear forms.
\subsection{Waring rank of $\det_3$}

Theorem \ref{thm: waring rank and border rank det3} will be a consequence of Theorem \ref{det18thm} and Theorem \ref{det17thm} below.

\begin{theorem}\label{det18thm}
 The Waring rank of $\det_3$ is at most $18$: $\bfR_S(\det_3) \leq 18$.
\end{theorem}
\begin{proof}
We give the rank $18$ decomposition for $\det_3$ explicitly, as a collection of $18$ linear forms on $\bbC^9 = \bbC^3 \otimes \bbC^3$ whose cubes add up to $\det_3$. The linear forms are given in coordinates recorded in the matrices below: the $3 \times 3$ matrix $( \zeta_{ij}) $ represents the linear form $\sum_{ij} \zeta_{ij} x_{ij}$. This presentation highlights some of the symmetries of the decomposition.

Let $\theta = \exp(2 \pi i / 6)$  and let $\bar{\theta}$ be its inverse. The tensor $\det_3 = T_{skewcw,2}^{\boxtimes 2} = \det (x_{ij})  \in S^3 (\bbC^3 \otimes \bbC^3)$ satisfies 
\[
 \det_3 = \sum _1^{18} L_i^3
\]
where $L_1 \vvirg L_{18}$ are the $18$ linear forms given by the following coordinates:
\begin{align*}
\begin{array}{lll}
L_1 = \begin{pmatrix}
-\theta & 0 & 0 \\
0 & -\frac{1}{3} & 0 \\
0 & 0 & \bar{\theta}
\end{pmatrix}
&L_2 =
\begin{pmatrix}
-\bar{\theta} & 0 & 0 \\
0 & -\frac{1}{3} & 0 \\
0 & 0 &\theta
\end{pmatrix}
&L_3 = 
\begin{pmatrix}
-\bar{\theta} & 0 & 0 \\
0 &  \frac{1}{3} \bar{\theta} & 0 \\
0 & 0 & \bar{\theta}
\end{pmatrix}\\ 
 L_4 = 
\begin{pmatrix}
-1 & 0 & 0 \\
0 & 0 & -\bar{\theta} \\
0 & -\frac{1}{3}\theta & 0
\end{pmatrix}
& L_5 = 
\begin{pmatrix}
\bar{\theta} & 0 & 0 \\
0 & 0 & 1 \\
0 & -\frac{1}{3}\theta & 0
\end{pmatrix} 
& L_6 = 
\begin{pmatrix}
\theta & 0 & 0 \\
0 & 0 & -\theta \\
0 & -\frac{1}{3}\theta & 0
\end{pmatrix}\\
L_7 = 
\begin{pmatrix}
0 &  \frac{1}{3} \bar{\theta} & 0 \\
-\theta & 0 & 0 \\
0 & 0 & 1
\end{pmatrix} 
&L_8 = 
\begin{pmatrix}
0 &  \frac{1}{3} \bar{\theta} & 0 \\
-\bar{\theta} & 0 & 0 \\
0 & 0 & -\bar{\theta}
\end{pmatrix}
&L_9 = 
\begin{pmatrix}
0 & \frac{1}{3}\theta & 0 \\
-\bar{\theta} & 0 & 0 \\
0 & 0 & 1
\end{pmatrix} \\ 
L_{10} = 
\begin{pmatrix}
0 & -\frac{1}{3}\theta & 0 \\
0 & 0 & \bar{\theta} \\
-1 & 0 & 0
\end{pmatrix} 
&L_{11} = 
\begin{pmatrix}
0 & -\frac{1}{3} \bar{\theta} & 0 \\
0 & 0 &\theta \\
-1 & 0 & 0
\end{pmatrix} 
&L_{12} = 
\begin{pmatrix}
0 & \frac{1}{3} & 0 \\
0 & 0 & -1 \\
-1 & 0 & 0
\end{pmatrix}\\ 
L_{13} = 
\begin{pmatrix}
0 & 0 & 1 \\
-1 & 0 & 0 \\
0 & -\frac{1}{3} & 0
\end{pmatrix} 
&L_{14} = 
\begin{pmatrix}
0 & 0 & 1 \\
\bar{\theta} & 0 & 0 \\
0 & \frac{1}{3}\theta & 0
\end{pmatrix} 
&L_{15} = 
\begin{pmatrix}
0 & 0 & 1 \\
\theta & 0 & 0 \\
0 &  \frac{1}{3} \bar{\theta} & 0
\end{pmatrix} \\ 
L_{16} = 
\begin{pmatrix}
0 & 0 & \bar{\theta} \\
0 & -\frac{1}{3}\theta & 0 \\
1 & 0 & 0
\end{pmatrix} 
&L_{17} = 
\begin{pmatrix}
0 & 0 & \bar{\theta} \\
0 & -\frac{1}{3} \bar{\theta} & 0 \\
-\bar{\theta} & 0 & 0
\end{pmatrix} 
&L_{18} = 
\begin{pmatrix}
0 & 0 &\theta \\
0 & -\frac{1}{3} \bar{\theta} & 0 \\
1 & 0 & 0
\end{pmatrix}
\end{array}.
\end{align*}

The equality can be verified by hand. A Macaulay2 file performing the calculation is available in Appendix B of the Supplementary Material.
 
\end{proof}
The Waring decomposition of \ref{det18thm} was generalized in \cite{JohnsTeit:ImprovedUpperBound} giving an upper bound for the Waring rank of the determinant polynomial $\det_m$.

\subsection{Waring border rank of $\det_3$}
The statement for the border rank is given by the following result. As in the previous proof, the border rank upper bound is proved explicitly giving linear forms, depending on a parameter $t$, whose cubes provide a border rank expression for $\det_3$. The algebraic numbers involved are more complicated than in the previous case. 

The result was achieved by numerical methods, which allowed us to sparsify the decomposition and ultimately determine the value of the coefficients. A detailed explanation of the method is given in Section \ref{subsec: discussion numerics}.

\begin{theorem}\label{det17thm}
The border Waring rank of $\det_3$ is at most $17$: $\uR_S(\det_3) \leq 17$. 
\end{theorem}
\begin{proof}
The $17$ linear forms providing a border rank decomposition of $\det_3$ are described below. Consider
\begin{align*}
\begin{array}{lll}
L_1(t)= \begin{pmatrix}
z_{1} & 0 & 0 \\
0 & z_{2} t & 0 \\
-1 & 0 & 0
\end{pmatrix} & 
L_{2}(t)= \begin{pmatrix}
z_{3} & 0 & 0 \\
z_{4} & 0 & z_{5} t \\
z_{6} & 0 & 0
\end{pmatrix} & 
L_3(t)= \begin{pmatrix}
-z_{36} & z_{7} t & 0 \\
-z_{38} & 0 & -z_{39} t \\
0 & 0 & t
\end{pmatrix} \\
L_4(t)= \begin{pmatrix}
0 & 0 & t \\
-z_{34} & 0 & 0 \\
0 & z_{8} t & -z_{35} t
\end{pmatrix} & 
L_{5}(t)= \begin{pmatrix}
0 & -z_{19} t & -z_{20} t \\
0 & 0 & 0 \\
-1 & 0 & 0
\end{pmatrix} & 
L_{6}(t)= \begin{pmatrix}
-z_{22} & z_{9} t & 0 \\
-z_{23} & 0 & -z_{24} t \\
-z_{25} & 0 & 0
\end{pmatrix} \\
L_{7}(t)= \begin{pmatrix}
z_{10} & z_{11} t & 0 \\
z_{12} & 0 & z_{13} t \\
z_{14} & 0 & 0
\end{pmatrix} & 
L_{8}(t)= \begin{pmatrix}
z_{15} & -t & 0 \\
z_{16} & 0 & z_{17} t \\
z_{18} & 0 & 0
\end{pmatrix} & 
L_{9}(t)= \begin{pmatrix}
0 & z_{19} t & z_{20} t \\
0 & z_{21} t & 0 \\
1 & 0 & 0
\end{pmatrix} \\
L_{10}(t)= \begin{pmatrix}
-z_{41} & 0 & 0 \\
0 & 0 & 0 \\
-z_{44} & 0 & 0
\end{pmatrix} & 
L_{11}(t)= \begin{pmatrix}
z_{22} & 0 & 0 \\
z_{23} & 0 & z_{24} t \\
z_{25} & 0 & 0
\end{pmatrix} & 
L_{12}(t)= \begin{pmatrix}
-z_{31} & z_{26} t & 0 \\
0 & z_{27} t & 0 \\
0 & 0 & t
\end{pmatrix} \\
L_{13}(t)= \begin{pmatrix}
z_{28} & z_{29} t & 0 \\
z_{30} & 0 & -t \\
0 & t & 0
\end{pmatrix} & 
L_{14}(t)= \begin{pmatrix}
z_{31} & z_{32} t & 0 \\
0 & 0 & 0 \\
0 & z_{33} t & -t
\end{pmatrix} & 
L_{15}(t)= \begin{pmatrix}
0 & 0 & -t \\
z_{34} & 0 & 0 \\
0 & 0 & z_{35} t
\end{pmatrix} \\
\multicolumn{3}{c}{
L_{16}(t)= \begin{pmatrix}
z_{36} & z_{37} t & 0 \\
z_{38} & 0 & z_{39} t \\
0 & z_{40} t & -t
\end{pmatrix}  \qquad 
L_{17}(t)= \begin{pmatrix}
z_{41} & z_{42} t & 0 \\
0 & z_{43} t & 0 \\
z_{44} & 0 & 0
\end{pmatrix}
}
\end{array}
\end{align*}

The coefficients $z_1 \vvirg z_{44}$ are algebraic numbers described as follows. Let $y_*$ be a real root of the polynomial 
\begin{align*}
& x^{27} - 2x^{26} + 17 x^{25} - 29 x^{24} + 81 x^{23} + 52 x^{22} - 726 x^{21} + 3451 x^{20} -10901 x^{19} + 25738 x^{18} -  \\
 &    50663 x^{17} + 72133 x^{16} - 72973 x^{15} + 10444 x^{14} + 138860 x^{13} - 308611 x^{12} + 427344 x^{11} \\
  &   - 267416 x^{10} - 196096 x^9 + 762736 x^8 - 1236736 x^7 + 1092352 x^6 - 537600 x^5 - 42240 x^4  + \\& 684032 x^3 - 1136640 x^2 + 1146880 x - 520192.\\
\end{align*}
For $i = 1 \vvirg 44$, we consider algebraic numbers $y_j$ in the field extension $\bbQ[y_*]$, described as a polynomial of degree (at most) $26$ in $y_*$ with rational coefficients. Notice that all the $y_j$'s are real. The expressions of the $y_1 \vvirg y_{44}$ in terms of $y_*$ are provided in the file \texttt{yy_exps} in Appendix C of the Supplementary Material. Let $z_j$ be the unique real cubic root of $y_j$. 

We are going to prove that, with this choice of coefficients $z_j$,
\begin{equation}\label{br17definingeqs}
  t^2 \det_3 + O(t^3) = \sum_{i=1}^{17} L_i(t)^{3} .
\end{equation}

The condition $t^2 \det_3 + O(t^3) = \sum_{i=1}^{17} L_i(t)^{3}$ is equivalent to the fact that the degree $0$ and the degree $1$ components of $\sum_{i=1}^{17} L_i(t)^{3} $ vanish and that the degree $2$ component equals $\det_3$. Given the sparse structure of the $L_i(t)$, this reduces to a   system of $54$ cubic equations in the $44$ unknowns $z_1 \vvirg z_{44}$. Our goal is to show that the algebraic numbers described above are a solution of this system.

We show that the $z_i$'s satisfy each equation as follows. After evaluating the equations at the $z_i$'s, there are two possible cases
\begin{enumerate}
 \item all monomials appearing in the equation are elements of $\bbQ[y_*]$; we say that this is an equation of type 1; there are $14$ such equations;
 \item at least one monomial appearing in the equation is not an element of $\bbQ[y_*]$; we say that this is an equation of type 2; there are $40$ such equations.
\end{enumerate}

For equations of type 1, we provide expressions of each monomial in terms of $y_*$. To verify that each expression is indeed equal to the corresponding monomial, it suffices to compare the cube of the given expression and the expression obtained by evaluating the monomial at the $y_j$'s. Finally, the equation can be verified in $\bbQ[y_*]$. This is performed by the file \texttt{checkingType1eqns.m2}.

For equations of type 2, let $u$ be one of the monomials which do not belong to $\bbQ[y_*]$. We claim that it is possible to choose the monomial in such a way that $\bbQ[u^3] = \bbQ[y_*]$. For each equation, we choose one of the monomials and we verify the claim as follows. The element $u^3$ has an expression in terms of $y_*$ which equals the chosen monomial evaluated at the $y_i$'s. Let $M_u$ be the $27 \times 27$ matrix with rational entries such that
\[
(1,u^3, \cdots , u^{3 \cdot 26})  = \left( 1 , y_* \vvirg y_*^{26} \right) \cdot M_u ;
\]
$M_u$ can be computed directly by considering the expressions of the powers of $u^3$ in terms of $y_*$. Then $\bbQ[u^3] = \bbQ[y_*]$ if and only if $M_u$ is full rank.

In particular $y_*$ has an expression in terms of $u^3$, which can be computed inverting the matrix $M_u$. A consequence of this is that $\bbQ[u] = \bbQ[y_*,u]$.

At this point, we observe that $\bbQ[u]$ contains the other monomials occurring in the equation as well. To see this, we proceed as in the case of equations of type 1. For each monomial occurring in the equation, we provide an expression in terms of $u$ (in fact, to speed up the calculation, we provide an expression in terms of $u$ and $y_*$, which is equivalent to an expression in $u$ because $\bbQ[u^3] = \bbQ[y_*]$ and $y_*$ has a unique expression in terms of $u^3$); we compare the cube of this expression (appropriately reduced modulo the minimal polynomial of $y_*$ and the relation between $u^3$ and $y_*$) with the expression obtained by evaluating the monomial at the $y_i$'s (expressed in terms of $y_*$). This shows that all monomials occurring in the expression belong to $\bbQ[u]$, and verifies that the given expressions are indeed equal to the corresponding monomials. Finally, the equation is verified in $\bbQ[u]$ as in the case of type 1. This is performed by the file \texttt{checkingType2eqns.m2}.
\end{proof}

\subsection{Discussion of how the decomposition was obtained}\label{subsec: discussion numerics}
 Many steps were accomplished by finding solutions of polynomial equations by 
nonlinear
optimization. In each case, this was accomplished using a variant of Newton's
method applied to the mapping of variable values to corresponding polynomial
values. 
The result of this procedure in each case is limited precision machine floating
point numbers. 

First, we  attempted   to solve  the equations describing
a Waring rank 17 decomposition of $\tdet_3$ with nonlinear optimization, 
namely, $\tdet_3 = \sum_{i=1}^{17}
(w_i')^{\ot 3}$, where $w_i' \in \CC^{3\times 3}$.  Instead of finding a
solution to working precision, we obtained a sequence of local
refinements to an approximate solution where the distance between $\det_3$ and its approximation is slowly
converging to zero, and some of the parameter values are exploding to
infinity. Numerically,  these are  Waring decompositions of   polynomials   
very close
to $\tdet_3$.

Next, this approximate solution  needed  to be upgraded to a solution to
equation~(\ref{br17definingeqs}). 
 
 We found     a choice of parameters in the neighborhood of a
solution,  and then applied   local optimization to solve to working
precision. We used the following method:  Consider the linear mapping $M :
\CC^{17} \to S^3(\CC^{3\times 3})$, $M(e_i) = (w_i')^{\ot 3}$, and  let $M =
U\Sigma V^\ast$ be its singular value decomposition (with respect to the
standard inner products for the natural coordinate systems).  We observed that 
the singular values  seemed  to be naturally partitioned by order of
magnitude.  We  estimated  this magnitude factor as $t_0 \approx 10^{-3}$, and
 wrote  $\Sigma'$ as $\Sigma$ where  we   multiplied each singular value by
$(t/t_0)^k$, with $k$ chosen to agree with this observed partitioning, so that
the constants remaining were reasonably sized. Finally,  we let   $M' = U 
\Sigma'
V^\ast$, which has entries in $\CC[[t]]$. Thus $M'$ is a representation of the
map $M$  with  a  parameter $t$.

 Next,  for each $i$, we  optimized   to find a best fit to the equation 
$(a_i+tb_i
 +t^2c_i)^{\ot 3} = M'(e_i)$, which is defined by polynomial equations in  the  
 entries
of $a_i$,  $b_i$ and  $c_i$.  The $a_i$,  $b_i$  and $c_i$  we   constructed in
this way proved to be a good initial guess to optimize 
equation~(\ref{br17definingeqs}), and we immediately saw quadratic convergence 
to a
solution to machine precision. At this point, we greedily  sparsified   the 
solution
by speculatively zero-ing values and re-optimizing, rolling back one step in 
case
of failure. After sparsification, it turned out the $c_i$ were not needed. The 
resulting matrices  are those given in the proof.

  To compute the minimal polynomials and
other integer relationships between quantities, we  used 
Lenstra-Lenstra-Lov\'asz integer lattice basis reduction \cite{MR682664}. As an
example, let $\zeta \in \RR$ be approximately an algebraic number of degree
$k$. Let $N$ be a large number inversely proportional to the error of
$\zeta$. Consider the integer lattice with
basis $\{e_i +
\lfloor N \zeta^i \rfloor e_{k+1}\} \subset \ZZ^{k+2}$, for $0\le i\le k$. 
Then elements of this
lattice are of the form  $v_0e_0 + \cdots + v_k e_k + E e_{k+1}$, where $E
\approx N p(\zeta)$, $p = v_0+v_1x+\cdots x_k x^k$.   Polynomials $p$ for
which $\zeta$ is an approximate root are distinguished by the property of
having relatively small Euclidean norm in this lattice. Computing a small norm
vector in an integer lattice is accomplished by LLL reduction of a known basis.

For example, 
 the fact that  the number field of degree 27 obtained by adjoining any
$z_\a^3$ to $\QQ$ contains all the rest    was determined via    LLL
reduction, looking for expressions of $z_\a^3$ as a polynomial in $z_\b^3$ for
some fixed $\b$. These expressions of $z_\a^3$ in a common number field can be
checked to have the correct minimal polynomial, and thus agree with our initial
description of the $z_\a$. LLL reduction was also 
used   to find the expressions of values as polynomials in the primitive
root of the various number fields.

After refining the known value of the parameters to $10,000$ bits of precision
using Newton's method, LLL reduction was successful in identifying the minimal
polynomials. The degrees were simply guessed, and the results checked by
evaluating the computed polynomials in the parameters to higher precision.


\begin{remark} With the minimal polynomial information, it is possible to check
that equation~(\ref{br17definingeqs}) is satisfied to any desired precision by 
the
parameters.
\end{remark}

\section{Tight Tensors in $\BC^3\ot \BC^3\ot \BC^3$}\label{tightsect}
 Following an analysis started in \cite{CGLVW}, we consider Kronecker squares of \emph{tight} tensors in $\bbC^3 \otimes \bbC^3 \otimes \bbC^3$. We compute their symmetry groups and numerically provide bounds to their tensor rank and border rank, highlighting the submultiplicativity properties.
 
 We refer to \cite{Strassen:AlgebraComplexity,gs005,BCS,CGLVW} for an exposition of the role of tightness in Strassen's work and in the laser method. In Lemma \ref{lemma: Tcw tight} below, we explicitly show that $T_{cw,q}$ and $T_{CW,q}$ are tight tensors. This fact was known and appears implicitly in \cite{gs005,DBLP:journals/corr/abs-1812-06952} and other related works: however we are not aware of a reference where the proof is given in its entirity.
 
\subsection{Tight tensors}
Recall the map $\Phi: GL(A) \times GL(B) \times GL(C) \to GL(A \otimes B \otimes C)$ from Section \ref{subsec: symmetry group of tensors} defining the action of $GL(A) \times GL(B) \times GL(C)$ on $A \otimes B \otimes C$. Its differential $d\Phi$ defines a map at the level of Lie algebras, mapping $\frakgl(A) \oplus \frakgl(B) \oplus \frakgl(C)$ to a subalgebra of $\frakgl(A \otimes B \otimes C)$. This subalgebrais  isomorphic to $(\frakgl(A) \oplus \frakgl(B) \oplus \frakgl(C)) / \BC^2$ where $\bbC^2 \simeq \ker d\Phi = \{ (\lambda_A \Id_A,\lambda_B \lambda \Id_B, \lambda_C \lambda \Id_C ) : \lambda_A + \lambda_B + \lambda_C = 0\}$ is the Lie algebra of the $2$-dimensional kernel of $\Phi$. Write $\frakg_T \subseteq \frakgl(A) \oplus \frakgl(B) \oplus  \frakgl(C)$ for the annihilator of $T$ under this action.

A tensor $T\in A\ot B\ot C$ is {\it tight} if $\frakg_T / \bbC^2$ contains a regular semisimple element. Given a basis $\{a_i : i = 1 \vvirg \dim A\}$ of $A$ and similarly for $B$ and $C$, write $T_{ijk}$ for the coordinates of a tensor $T$ in the induced basis $\{a_i \otimes b_j \otimes c_k\}$ of $A \otimes B \otimes C$. The support of a tensor $T \in A \otimes B \otimes C$ is 
\[
\supp(T) = \{ (i,j,k) : T_{ijk} \neq 0\}.                                                                                                                                                                                                                                                                                                                                                                               
\]
Tightness can be defined combinatorially with respect to a basis, see, e.g., \cite[Def. 1.3]{CGLVW}. Explicitly, $T$ is tight if and only if there exist bases of $A,B,C$ and injective functions $s_A : \{ 1 \vvirg \dim A\} \to \bbZ $, $s_B : \{1 \vvirg \dim B \} \to \bbZ$, $s_C : \{1 \vvirg \dim C\} \to \bbZ$ such that 
\[
 s_A(i) + s_B(j) + s_C(k) = 0 \quad \text{ for every $(i,j,k) \in \supp(T)$}.
\]
The following result was ``known to the experts'' but since we do not have a reference for it, we provide its proof. 
\begin{lemma}\label{lemma: Tcw tight}
 The tensors $T_{cw,q}$ and $T_{CW,q}$ are tight.
\end{lemma}
\begin{proof}

Write $q = 2u$ or $ q = 2u+1$ depending on the parity of $q$. Consider the change of basis 
\begin{align*}
 a_0 &\mapsto a_0 \\
 a_j &\mapsto \frac{\sqrt{2}}{2} (a_j + a_{u+j}) \quad \text{ for $j = 1 \vvirg u$}\\
 a_{u+j} &\mapsto \frac{\sqrt{-2}}{2}(a_{j} - a_{u+j}) \quad \text{ for $j =1  \vvirg u$}\\
 a_q &\mapsto a_q \quad \text{(if $q$ is odd)}\\
 a_{q+1} &\mapsto a_{q+1}
\end{align*}
and similarly on $B$ and $C$.

After this change of basis, regarding $T_{cw,q}$ and $T_{CW,q}$ as symmetric tensors in $S^3 A$, we have
\begin{align*}
 T_{cw,2u} &= a_0 \left( \textsum_{j=1}^u a_j a_{u+j} \right), \\
 T_{CW,2u} &= a_0 \left(\textsum_{j=1}^u a_j a_{u+j} \right) + a_0^2a_{q+1},\\
 \end{align*}
 or
 \begin{align*}
 T_{cw,2u+1} &= a_0 \left( \textsum_{j=1}^u a_j a_{u+j} + a_q^2\right), \\
 T_{CW,2u+1} &= a_0 \left(\textsum_{j=1}^u a_j a_{u+j} + a_q^2 \right) + a_0^2a_{q+1},
\end{align*}
depending on the parity of $q$.

Define $s = s_A=s_B=s_C$ by
\begin{align*}
 s(0) &= 2, \\
s(j) &= 2+j \quad \text{for $j=1 \vvirg u$},\\
s(u+j) &= -j-4 \quad \text{for $j=1 \vvirg u$},\\
s(q)&=-1 \quad \text{if $q$ is odd},\\
s(q+1) &= -4.
\end{align*}
It is easy to verify that $s(i)+s(j) + s(k) = 0$ if $ (i,j,k) \in \supp(T_{CW,q})$. Moreover, since $\supp(T_{cw,q}) \subseteq \supp(T_{CW,q})$, the same holds for $(i,j,k) \in \supp(T_{cw,q})$. This concludes the proof.
\end{proof}

The combinatorial characterization of tightness makes it clear that this property only depends on the support of a tensor in a given basis; we say that a support $\calS$ is tight if every tensor having support $\calS$ is tight.

A tensor $T \in A \otimes B \otimes C$ is {\it concise} if the induced linear maps $T_A : A^* \to B \otimes C$, $T_B : B^* \to A \otimes C$, $T_C: C^* \to A \otimes B$ are injective. We say that a concise tensor $T \in \bbC^m \otimes \bbC^m \otimes \bbC^m$ has {\it minimal rank} (resp. {\it minimal border rank}) if $\bfR(T) = m$ (resp. $\uR(T)=m$).

Given concise tensors $T_1 \in A_1 \otimes B_1 \otimes C_1$ and $T_2 \in A_2 \otimes B_2 \otimes C_2$, \cite[Theorem 4.1]{CGLVW} shows that 
\begin{equation}\label{tboxs}
\frakg_{T_1 \boxtimes T_2} \supseteq \frakg_{T_1} \otimes \Id_{A_2 \otimes B_2 \otimes C_2} + \Id_{A_1 \otimes B_1 \otimes C_1} \otimes \frakg_{T_2};
\end{equation}
moreover if $\frakg_{T_1} = 0$ and $\frakg_{T_2} = 0$ then equality holds $\frakg_{T_1 \boxtimes T_2} = 0$.
  
The strict containment in \eqref{tboxs} occurs, for instance, in the case of the matrix multiplication tensor. In \cite{CGLVW}, we posed the problem of characterizing tensors $T \in A \otimes B \otimes C$ such that $\frakg_T \otimes \Id_{A \otimes B \otimes C} + \Id_{A \otimes B \otimes C} \otimes \frakg_T$ is strictly contained in $\frakg_{T^{\boxtimes 2}} \subset  \frakgl(A^{\otimes 2}) +\frakgl( B^{\otimes 2}) +\frakgl(C^{\otimes 2})$.

Proposition \ref{prop: table on tight tensors} provides several additional examples of tensors in $\bbC^3 \otimes \bbC^3 \otimes \bbC^3$ for which this containment is strict.

\subsection{Tight supports in $\BC^3\ot \BC^3\ot \BC^3$}

From \cite[Proposition 2.14]{CGLVW}, one obtains an exhaustive list of unextendable tight supports for tensors in $\BC^3\ot \BC^3\ot \BC^3$, up to the action of $\bbZ_2 \times \frakS_3$, where $\frakS_3$ acts permuting the factors and $\bbZ_2$ acts by reversing the order of the basis elements. In fact, tightness is invariant under the action of the full $\frakS_3$ acting by permutation on the basis vectors. This additional simplification, pointed out by J. Hauenstein, provides the following list of $9$ unextendable tight supports up to the action of $((\frakS_3)^{\times 3}) \rtimes \frakS_3$.

\begin{align*}
\calT_{1 } &= \{{(1,1,3), (1,2,2), (2,1,2), (3,3,1)}\}; \\ 
\calT_{2 } &= \{{(1,1,3), (1,3,2), (2,3,1), (3,2,2)}\}; \\ 
\calT_{3 } &= \{{(1,1,3), (1,2,2), (1,3,1), (2,1,2), (3,2,1)}\}; \\ 
\calT_{4 } &= \{{(1,1,3), (1,2,2), (2,1,2), (2,3,1), (3,2,1)}\}; \\ 
\calT_{5 } &= \{{(1,1,3), (1,2,2), (2,3,1), (3,1,2), (3,2,1)}\}; \\ 
\calT_{6 } &= \{{(1,1,3), (1,3,2), (2,2,2), (3,1,2), (3,3,1)}\}; \\ 
\calT_{7 } &= \{{(1,1,3), (1,2,2), (1,3,1), (2,1,2), (2,2,1), (3,1,1)}\}; \\ 
\calT_{8 } &= \{{(1,1,3), (1,3,2), (2,2,2), (2,3,1), (3,1,2), (3,2,1)}\}; \\ 
\calT_{9 } &= \{{(1,2,3), (1,3,2), (2,1,3), (2,2,2), (2,3,1), (3,1,2), (3,2,1)}\}; \\ 
\end{align*}
Supports $\calS_2$ and $\calS_3$ of \cite{CGLVW} are equivalent to support $\calS_1= \calT_1$; supports $\calS_8$ and $\calS_{10}$ are equivalent to support $\calS_6 = \calT_4$.

The following result characterizes tight tensors in $\bbC^3 \otimes \bbC^3 \otimes \bbC^3$ up to isomorphism.
\begin{proposition}
Let $T \in \bbC^3 \otimes \bbC^3 \otimes \bbC^3$ be a tight tensor with unextendable tight support in some basis. Then, up to permuting the three factors, $T$ is isomorphic to exactly one of the following.
\begin{align*}
T_{1}:=&a_1\ot b_1\ot c_3+ a_1\ot b_2\ot c_2+ a_2\ot b_1\ot c_2+a_3\ot b_3\ot 
c_1
\\
T_{2}:=&a_1\ot b_1\ot c_3+ a_1\ot b_3\ot c_2+ a_2\ot b_3\ot c_1+ a_3\ot b_2\ot 
c_2
\\
T_{3}:=&a_1\ot b_1\ot c_3+ a_1\ot b_2\ot c_2+ a_1\ot b_3\ot c_1+ a_2\ot b_1\ot 
c_2+ a_3\ot b_2\ot c_1
\\
T_{4}:=&a_1\ot b_1\ot c_3+ a_1\ot b_2\ot c_2+ a_2\ot b_1\ot c_2+ a_2\ot b_3\ot 
c_1+ a_3\ot b_2\ot c_1
\\
T_{5}:=&a_1\ot b_1\ot c_3+ a_1\ot b_2\ot c_2+ a_2\ot b_3\ot c_1+ a_3\ot b_1\ot 
c_2+ a_3\ot b_2\ot c_1
\\
T_{6}:=&a_1\ot b_1\ot c_3+ a_1\ot b_3\ot c_2+ a_2\ot b_2\ot c_2+ a_3\ot b_1\ot 
c_2+ a_3\ot b_3\ot c_1
\\
T_{7}:= &a_1 \ot b_1 \ot c_3 + a_1 \ot b_2 \ot c_2 + a_1 \ot b_3 \ot c_1 + a_2 \ot b_1 \ot c_2 + a_2 \ot b_2 \ot c_1 + a_3 \ot b_1 \ot c_1
\\
T_{8}:=&a_1\ot b_1\ot c_3+ a_1\ot b_3\ot c_2+ a_2\ot b_2\ot c_2+ a_2\ot 
b_3\ot 
c_1+ a_3\ot b_1\ot c_2+ a_3\ot b_2\ot c_1
\\
T_{9,\mu}:=&a_1\ot b_2\ot c_3+ a_1\ot b_3\ot c_2+ a_2\ot b_1\ot c_3+ a_2\ot 
b_2\ot c_2+ a_2\ot b_3\ot c_1 +   a_3\ot b_1\ot c_2 \\ ~ & \hfill +  \mu \cdot a_3\ot b_2\ot c_1 \quad \text{for some $\mu \in \bbC \setminus \{0\}$}.
\\
\end{align*}
\end{proposition}
\begin{proof}
 The result of \cite[Proposition 2.14]{CGLVW} and the discussion above shows that $T$ is, up to permutation of the factors, equivalent to a tensor with support $\calT_i$ for some $i = 1 \vvirg 9$.

For $i = 1 \vvirg 8$, it is straightforward to verify that all tensors with support $\calT_i$ are isomorphic, via the change of bases given by three diagonal matrices.

The case of $\calT_{9}$ is slightly more involved but essentially the same argument shows that a tensor $T$ with support $\calT_9$ is isomorphic to $T_{9,\mu}$, for some $\mu$.

Finally, we have to show that any two of the tensors in the statement are not isomorphic. For tensors having distinct supports, this is a consequence of Proposition \ref{prop: table on tight tensors} below: indeed, if $T,T'$ are two of the tensors above, Proposition \ref{prop: table on tight tensors} shows that either $\dim \frakg_T \neq \dim \frakg_{T'}$ or $\dim \frakg_{T^{\boxtimes 2}} \neq \dim \frakg_{{T'}^{\boxtimes 2}}$.

As for the tensors with support $\calT_9$, we proceed as follows. Let $T = T_{9,\mu}$ and $T' = T_{9,\mu'}$ with $\mu \neq \mu'$. We show that $T$ is not isomorphic to $T'$. Suppose by contradiction that there is a triple of $3 \times 3$ matrices $g = (g_A,g_B,g_C) \in GL_3 \times GL_3 \times GL_3$ with $g (T) = T'$. One sees that  in each case,   $g_A,g_B,g_C$ have to be diagonal matrices, and an explicit calculation shows that there is no triple of diagonal matrices such that $g (T) = T'$.
\end{proof}

We point out that $T_{7}$ is isomorphic to the Coppersmith-Winograd tensor $T_{CW,1}$, as well as to the structure tensor of the algebra $\bbC[x]/(x^3)$.

The tensors $T_{cw,2}$ and $T_{skewcw,2}$ are degenerations of $T_{9,\mu}$, respectively for $\mu=1$ and $\mu = -1$. In particular, they do not have an unextendable tight support in some basis.

\begin{proposition} \label{prop: table on tight tensors}
For $i = 1 \vvirg 9$, the following table records $\dim \frakg_{T_i}$ and $\dim \frakg_{T_i^{\boxtimes 2}}$.
\[
\begin{array}{lccr}\toprule 
T & \dim\fg_T & \dim \fg_{T^{\boxtimes 2}} \\
\midrule 
T_{1} & 5 & 22  \\
T_{2} & 3 & 9  \\
T_{3} & 5 & 13 \\
T_{4} & 4 & 9 \\
T_{5} & 3 & 7 \\
T_{6}  & 2 & 5  \\
T_{7}    & 6 & 28 \\
T_{8}  & 1 & 2 \\
T_{9,-1}& 5 & 10   \\
T_{9,\mu} \ \ {\scriptsize \text{(for $\mu \neq 0, -1$)}}& 1 & 2 \\ 
\end{array} 
\]
In summary 
\[
\dim \frakg_{T^{\boxtimes 2}} > 2\dim \frakg_T 
\]
for tight tensors in $\bbC^3 \otimes \bbC^3 \otimes \bbC^3$ with unextendable tight supports $\calT_1 \vvirg \calT_7$.
\end{proposition}
\begin{proof}
 For $T_1 \vvirg T_8$ and for the $T_{9,-1}$, the proof follows by a direct calculation. The first part of the Macaulay2 file \texttt{symmetryTightSupports.m2} in Appendix E of the Supplementary Material computes the dimension of the symmetry algebras of interest in these cases.
 
 The second part of the file deals with the case $T_{9,\mu}$ when $\mu \neq -1$. By tightness, $\dim \frakg_{T_{9,\mu}} \geq 1$.
 
 Consider the linear map $\omega_{T_{9,\mu}} : \frakgl(A) + \frakgl(B) + \frakgl(C) \to A \otimes B \otimes C$ defined by $(X,Y,Z) \mapsto (X,Y,Z).T_{9,\mu}$. Then $\frakg_{T_{9,\mu}} = [\ker (\omega_{T_{9,\mu}})] / \bbC^2$, where $\bbC^2$ corresponds to $\ker d\Phi$.
 
 The second part of the file \texttt{symmetryTightSupports.m2} computes a matrix representation of $\omega_{T_{9,\mu}} $, depending on a parameter $\mu$ (\texttt{t} in the file). Let $F_\mu$ be this $27 \times 27$ matrix representation. Then, it suffices to select a $24 \times 24$ submatrix whose determinant is a nonzero univariate polynomial in $\mu$. If $\mu$ is a value for which $\dim \frakg_{T_{9,\mu}} > 1$, then $\mu$ has to be a root of this univariate polynomial.
 
 In the example computed in the file, we select a $24 \times 24$ submatrix whose determinant is $(\mu+1)^6\mu$, showing that the only possible values of $\mu$ for which $\dim \frakg_{T_{9,\mu}} > 1$ are $\mu = 0$ or $\mu=-1$. The case $\mu = -1$ was considered separately. The case $\mu = 0$ does not correspond to a unextendable support, so it is not of interest. We point out that however, $\rank(\omega_{T_{9,0}}) = 24$, namely $\dim \frakg_{T_{9,0}} =1$.
 
For $T_{9,\mu}^{\boxtimes 2}$, we follow essentially the same argument. By tightness, and \eqref{tboxs}, we obtain $\dim \frakg_{T_{9,\mu}^{\boxtimes 2}} \geq 2$. The third part of \texttt{symmetryTightSupports.m2} computes a matrix representation of the map $\omega_{T_{9,\mu}^{\boxtimes 2}}$, depending on a parameter $\mu$: this is a $729 \times 243$ matrix of rank at most $239$. 

In the example computed in the file, we select a $239 \times 239$ submatrix whose determinant is the univariate polynomial $\mu^8(\mu+1)^{12}$. As before, we conclude.
\end{proof}

We also provide the values of the border rank of the tensors in $\bbC^3 \otimes \bbC^3 \otimes \bbC^3$ having unextendable tight support and numerical evidence for the values of border rank of their Kronecker square. They are recorded in the following table. The values of the border rank for the $T_i$'s are straightforward to verify. The lower bounds for the Kronecker squares are obtained via Koszul flattenings. In the cases labeled by N/A the upper bounds coincide with the multiplicative upper bound; in the other cases, the upper bound is obtained via numerical methods, and the last column of the table records the $\ell_2$ distance (in the given basis) between the tensor obtained via the numerical approximation and the Kronecker square. The numerical approximations are recorded in the supplementary files in Appendix F of the Supplementary Material.
\[
\begin{array}{lccc}\toprule 
T & \ur(T) & \ur(T^{\boxtimes 2}) & \ell_2 \text{ error for  upper bound in $T^{\boxtimes 2}$ decomposition} \\
\midrule 
T_{1} & 3 & 9 & N/A \\
T_{2} & 4 & [11,14] & 0.000155951 \\
T_{3} & 4 & [11,14] &  0.00517612 \\
T_{4} & 4 & 14 & 0.0144842 \\
T_{5} & 4 & [11,15] & 0.0237172 \\
T_{6} & 4 & [11,15] & 0.00951205 \\
T_{7} & 3 & 9 & N/A \\
T_{8} & 4 & [14,16] & N/A \\
T_{9, -1}& 5 & [16,19] & 0.0231353  \\
T_{9,\mu} \ \ {\scriptsize \text{(for $\mu \neq 0,-1$)}}& 4 & [15,16] & N/A \\
\end{array} 
\]

\section{A method to compute flattenings of structured tensors}\label{section: boxparametrized}
In this section, we explain how to compute the matrices $\Phi_1 \vvirg \Phi_4$ in Section \ref{secondpf} and the matrices $\Psi_1 \vvirg \Psi_8$ in Section \ref{subsec: proof kron cube Tcw}. 

The matrices $\Phi_1 \vvirg \Phi_4$ and $\Psi_1 \vvirg \Psi_8$ arise via a series of tensor contractions of highly structured tensors. In this section, we introduce the notion of \emph{box parametrized} sequence of tensors. Lemma \ref{lemma: box parametrized under product and contraction} below shows that contraction of box parametrized tensors gives rise to box parametrized tensors; in addition, the expression of the tensors resulting from the contraction is particularly easy to control.

We will then show that the tensors in Section \ref{secondpf} and Section \ref{subsec: proof kron cube Tcw} which give rise to the matrices $\Phi_1 \vvirg \Phi_4$ and $\Psi_1 \vvirg \Psi_8$ are box parametrized. This allows us to track down the entries of the final matrices as functions of the dimension $q$. 

The full calculation of the matrices is left to the scripts available in Appendix D of the Supplementary Material.

The point of view is partially inspired by the interpretation of tensors in communication models, where a tensor on $k$ factors is regarded as a function from $\underbrace{\bbN \ttimes \bbN}_{k} \to \bbC$ with finite support sending a $k$-tuple of integers to the corresponding coefficient of the tensor. Explicitly, for every $j = 1 \vvirg k$ fix a basis $\{v_i^{(j)}\}$ on the $j$-th factor: given a finite support $\Sigma \subseteq \bbN^{\times k}$, the tensor $T = \sum_{(i_1 \vvirg i_k) \in \Sigma} t_{i_1 \vvirg i_k} v_{i_1}^{(1)} \ootimes v_{i_k}^{(k)}$ corresponds to the function defined by $T(i_1 \vvirg i_k) = t_{i_1 \vvirg i_k}$. We do not explicitly write the dimensions of the factors.

Let $\calT = \{ T_q : q \in \bbN \}$ be a sequence of tensors of order $k$. We say that $\calT$ is \emph{basic box-parametrized} if, for every $q$
\[
 T_q = p(q) \sum_{(i_1 \vvirg i_k) \in \Sigma_q} v_{i_1}^{(1)} \ootimes v_{i_k}^{(k)}
\]
where $p(q)$ is a univariate polynomial in $q$ and the support $\Sigma_q$ is defined by conditions $\eta_j q + \theta_j \leq i_j \leq H_j q + \Theta_j$ for $\eta_j,H_j \in \{0,1\}$ and $\theta_j,\Theta_j \in \bbZ_{\geq 0}$, and any number (not depending on $q$) of equalities $i_j = i_{j'}$ among indices. Without loss of generality, assume that the inequalities are sharp for every $j$, in the sense that for every $i_j$ satisfying the $j$-th inequality, the basis element $v^{(j)}_{i_j}$ does appear in $T_q$. We often say that $\calT$ is basic box-parametrized for $q \geq q_0$ for some $q_0$, in the sense that the sequence has the desired structure for $q \geq q_0$.

\begin{example}\label{example: easy box parametrized}
 The sequence $T_q = v_0^{(1)} \otimes \sum_{i=1}^q v_{i}^{(2)} \otimes v_i^{(3)}$ is basic box-parametrized for $q \geq 1$, with support $\Sigma_q$ defined by the conditions 
 \[
0 \leq i_1 \leq 0, \quad 1 \leq i_2 \leq q, \quad 1 \leq i_2 \leq q, \quad i_2 = i_3.
 \]
\end{example}

We define a contraction operation between the $j_1$-th and the $j_2$-th factor of $\calT$, obtained by summing over the corresponding indices: in other words, the contraction is the image of $T$ via the trace map $\sum u_i^{(j_1)} \otimes u_i^{(j_2)}$ applied to the $j_1$-th and $j_2$-th factors, where $\{ u_i^{(j)} \}$ is the dual basis to the fixed basis $\{ v_i^{(j)} \}$ on the $j$-th factor.

\begin{lemma}\label{lemma: box parametrized under product and contraction}
 Let $\calT$, $\calT'$ be basic box-parametrized tensors for $q \geq q_0$ and $q \geq q_0'$ respectively. Then
 \begin{itemize}
  \item $\calT \otimes \calT'$ is basic box-parametrized for $q \geq \max\{ q_0,q_0'\}$;
  \item the contraction of $\calT$ on factors $j_1$ and $j_2$ is basic box-parametrized for $q \geq \max \{ | \theta_{j_1} - \theta_{j_2} | , | \Theta_{j_1} - \Theta_{j_2}|, q_0\}$; moreover, if the univariate coefficient $p(q)$ of $\calT$ is a polynomial of degree $e$, then the coefficient of the tensor resulting from the contraction has degree at most $e+1$.
 \end{itemize}
\end{lemma}
\begin{proof}
 The first statement is immediate.
 
For the second statement, without loss of generality assume $j_1 = 1$ and $j_2 = 2$. First observe that if $\calT$ is basic box-parametrized, then summing over the first index, or equivalently applying the linear map $\sum_i u_i^{(1)}$, generates a basic box-parametrized tensor; the coefficient of this tensor has the same degree as the coefficient of $\calT$ unless the first index $i_1$ is not related by equality to any other index, and $\eta_1 =0$ and $H_1 = 1$; in the latter case, the degree of the coefficient is increased by one.

Now, contraction of $\calT$ on factors $1$ and $2$ is equivalent to first imposing the equality $i_1 = i_2$ on the support $\Phi_q$ and then summing up on the first and second index. Imposing the equality $i_1 = i_2$ effects the inequalities of $i_1$ and $i_2$ as follows:
\[
 \max\{ \eta_1 q + \theta_1, \eta_2q + \theta_2\} \leq i_1 = i_2 \leq  \min \{ H_1 q + \Theta_1, H_2q + \Theta_2\}.
\]
Each of the two bounds can be replaced by one of the two linear functions (uniformly in $q$) whenever $q \geq\{ | \theta_{1} - \theta_{2} | , | \Theta_{1} - \Theta_{2}|\}$. This, together with the previous observation, concludes the proof.
\end{proof}

Given two sequences of tensors $\calT^{(1)} ,\calT^{(2)}$ of order $k$, we define their sum as $\calT_1 + \calT_2 = \{ T_q^{(1)} + T_q^{(2)} : q \in \bbN \}$. We say that a sequence $\calT$ is \emph{box parametrized} (for $q \geq q_0$) if $\calT$ is a finite sum of basic box-parametrized sequences of tensors (for $q \geq q_0$). Observe that a sequence of tensors with constant dimensions is box parametrized if and only if its coefficients are univariate polynomials in $q$.

We will show that the maps $\Phi_1 \vvirg \Phi_4$ in the proof of Theorem \ref{cwbadnews} in Section \ref{secondpf} and the maps $\Psi_1 \vvirg \Psi_8$ in the proof of Theorem \ref{cwbadnews3} in Section \ref{subsec: proof kron cube Tcw} are box parametrized.

The scripts in Appendix D perform the contraction of box parametrized tensors according to Lemma \ref{lemma: box parametrized under product and contraction}, keeping track of the univariate polynomial coefficients and of the lower bound $q_0$ for which the expressions are valid. The final result is that the maps $\Phi_1 \vvirg \Phi_4$ are box parametrized for $q \geq 5$ and the maps $\Psi_1 \vvirg \Psi_8$ are box parametrized for $q \geq 6$.

In the following, we show that the tensors involved in the various contractions are box parametrized. Lemma \ref{lemma: box parametrized under product and contraction} guarantees that the results of the contractions are box parametrized as well.

First, notice that $T_{cw,q}$ is box parametrized for $q \geq 1$, as it is the sum of three tensors as the ones described in Example \ref{example: easy box parametrized}. By Lemma \ref{lemma: box parametrized under product and contraction}, we deduce that $T_{cw,q}^{\otimes 2}$ (regarded as a tensor of order $6$) and $T_{cw,q}^{\otimes 3}$ (regarded as a tensor of order $9$) are box parametrized. In all three cases, the polynomials defining the coefficients have degree $0$.

\subsection{Restriction}

We show that the two restriction maps $\phi_2: A^{\otimes 2} \to \bbC^3$ and $\phi_3: A^{\otimes 3} \to \bbC^5$ are box parametrized as tensors of order $3$ and $4$ respectively.

Write $\phi_2 = X_0 \otimes e_0 + X_1 \otimes e_1  + X_2 \otimes e_2$, where $\bbC^3 = \langle e_0,e_1,e_2\rangle$ and $X_0,X_1,X_2 \in {A^{\otimes 2}}^*$. It suffices to show that $X_0,X_1,X_2$ are box parametrized, regarded as tensors of order two. Using a basis dual to the basis of $A^{\otimes 2}$, we have 
\begin{align*}
 X_0 &= \alpha_0 \otimes \alpha_1 + \alpha_1 \otimes \alpha_0 + \alpha_1 \otimes \alpha_1 \\
 X_1 &= \alpha_0 \otimes \textsum_1^q \alpha_i + \textsum_1^q \alpha_i \otimes \alpha_0 \\
 X_2 &= \alpha_0 \otimes \alpha_2 + \alpha_2 \otimes \alpha_0 + \alpha_2 \otimes \alpha_1 + \alpha_3 \otimes \alpha_3.
\end{align*}
This shows that $X_0,X_1,X_2$ are box parametrized. 

Similarly, write $\phi_3 = Y_0 \otimes e_0 + \cdots + Y_4 \otimes e_4$, where $\bbC^5 = \langle e_0 \vvirg e_4\rangle$ and $Y_0 \vvirg Y_4 \in {A^{\otimes 3}}^*$. Directly from the definition in Section \ref{subsec: proof kron cube Tcw}, it is immediate that $Y_0 \vvirg Y_4$ are box parametrized and therefore $\phi_3$ is box parametrized as well.

Applying Lemma \ref{lemma: box parametrized under product and contraction}, we deduce that the two sequences $\phi_2( T_{cw,q}^{\otimes 2})$ and $\phi_3( T_{cw,q}^{\otimes 3})$ are box parametrized.

\subsection{Koszul maps}

The Koszul differentials on $\bbC^3$ and $\bbC^5$ used in the definition of the Koszul flattenings are the skew-symmetric projections $\bbC^3 \otimes \bbC^3 \to \Lambda^2 \bbC^3$ and $\Lambda^2 \bbC^5 \otimes \bbC^5 \to \Lambda^3 \bbC^5$. They are both fixed size, therefore they are box parametrized.

By Lemma \ref{lemma: box parametrized under product and contraction}, we deduce that the resulting Koszul flattenings $(\phi_2(T_{cw,q}^{\boxtimes 2}) ) ^{\wedge 1}$ and $(\phi_3(T_{cw,q}^{\boxtimes 3}) ) ^{\wedge 2}$ are box parametrized, regarded as tensors of order $6$ and $8$ respectively.

\subsection{Diagonalizing maps}

Recall that the maps $\Phi_1 \vvirg \Phi_4$ in the proof of Theorem \ref{cwbadnews} and the maps $\Psi_1 \vvirg \Psi_8$ in the proof of Theorem \ref{cwbadnews3} are the restrictions of $(\phi_2(T_{cw,q}^{\boxtimes 2}) ) ^{\wedge 1}$ and $(\phi_3(T_{cw,q}^{\boxtimes 3}) ) ^{\wedge 2}$ to the multiplicity spaces of the isotypic components for the action of $\frakS_{q-3}$ and $\frakS_{q-5}$.

We analyze the square case in detail. For the square case, let $\calM$ be the matrix of change of basis on $\bbC^q$ from the basis $\{ e_1 \vvirg e_q\}$ to the basis $\{
 e_1,e_2,e_3, \sum_4^q e_i , e_5-e_4 \vvirg e_q - e_{q-1} \}$. Explicitly
 \[
  \calM = \left[\begin{array}{ccccc}
           \Id_3 &  & & & \\
            & 1 & 1 & \cdots & 1 \\
            & -1 & 1 & & \\
            & & \ddots & \ddots & \\
            & & & -1 & 1 
          \end{array}\right].
 \]
In particular, $\calM$ diagonalizes the action of $\frakS_{q-3}$ and therefore the change of basis defined by $\Id_{\bbC^3} \boxtimes \calM^{\boxtimes 2}$ on $\bbC^3 \otimes B^{\otimes 2}$ brings the matrix representing $(\phi_2(T_{cw,q}^{\boxtimes 2}) ) ^{\wedge 1}$ into a block diagonal matrix, whose diagonal blocks are matrices representing the maps $f_i : \bbC^3 \otimes (M_i \otimes W_i) \to \Lambda^2 \bbC^3 \otimes (M_i \otimes W_i)$ from \eqref{eq: diagonal blocks square}; denote the diagonal blocks by $f^\calM_1 \vvirg f^\calM_4$.

Because of our choice of basis, the multiplicity subspaces $\bbC^3 \otimes \langle w_i \rangle \otimes M_i$ and $\Lambda^2 \bbC^3 \otimes \langle w_i \rangle \otimes M_i$ described in Section \ref{secondpf} are spanned by basis vectors, so that the matrices representing $\Phi_1 \vvirg \Phi_4$ are given by submatrices of $f^\calM _1 \vvirg f^\calM_4$. More precisely, setting $\pi_{inv},\pi_{std}$ to be the matrices of the two coordinate projections of $\bbC^q$ onto $\langle e_1 \vvirg e_4\rangle$ and $\langle e_5 \rangle$, we have 
\begin{align*}
 \Phi_1 &= (\Id_{\Lambda^2 \bbC^3} \boxtimes \pi_{inv} \boxtimes \pi_{inv}) \circ f^\calM_1 \circ (\Id_{\bbC^3} \boxtimes {\pi_{inv} \boxtimes \pi_{inv}}) ^T, \\
 \Phi_2 &= (\Id_{\Lambda^2 \bbC^3} \boxtimes \pi_{inv} \boxtimes \pi_{std}) \circ f^\calM_2 \circ (\Id_{\bbC^3} \boxtimes {\pi_{inv} \boxtimes \pi_{std}}) ^T, \\
 \Phi_3 &= (\Id_{\Lambda^2 \bbC^3} \boxtimes \pi_{std} \boxtimes \pi_{inv}) \circ f^\calM_3 \circ (\Id_{\bbC^3} \boxtimes {\pi_{std} \boxtimes \pi_{inv}}) ^T, \\
 \Phi_4 &= (\Id_{\Lambda^2 \bbC^3} \boxtimes \pi_{std} \boxtimes \pi_{std}) \circ f^\calM_4 \circ (\Id_{\bbC^3} \boxtimes {\pi_{std} \boxtimes \pi_{std}}) ^T. \\
\end{align*}
Since the composition can be performed on the single factors, by Lemma \ref{lemma: box parametrized under product and contraction} it suffices to show that the four matrices $\calM^{-1} \circ \pi_{inv}^T$, $\calM^{-1} \circ \pi_{std}^T$, $\pi_{inv} \circ \calM$ and $\pi_{std} \circ \calM$ are box parametrized.

From the structure of $\calM$, it is clear that $\pi_{inv} \circ \calM$ and $\pi_{std} \circ \calM$ are box parametrized. The computation of $\calM^{-1}$ is straightforward, and it is easy to see that $\calM^{-1} \circ \pi_{inv}^T$, $\calM^{-1} \circ \pi_{std}^T$ are box parametrized.

This shows that $\Phi_1 \vvirg \Phi_4$ are box parametrized. The script available in Appendix D computes the box parametrized representation of $\Phi_1 \vvirg \Phi_4$ starting from the box parametrized version of $T_{cw}$, the restriction map $\phi_2$, the Koszul differential and the four matrices $\calM^{-1} \circ \pi_{inv}^T$, $\calM^{-1} \circ \pi_{std}^T$, $\pi_{inv} \circ \calM$ and $\pi_{std} \circ \calM$.

The cube case is similar. Now, restriction space $\bbC^3$ is a $\bbC^5$, the top left block in the matrix $\calM$ is a $5 \times 5$ identity block, the result of the conjugation by $\calM$ is block diagonal with $8$ blocks, corresponding to the eight isotypic components. The coordinate projections $\pi_{inv}$ and $\pi_{std}$ are onto $\langle e_1 \vvirg e_6 \rangle$ and $\langle e_7\rangle$. The script computes the box parametrized representation of the matrices $\Psi_1 \vvirg \Psi_8$.

\subsection*{Acknowledgements}
Landsberg supported by NSF grant AF-1814254. Gesmundo acknowledges support from VILLUM FONDEN via the QMATH Centre of Excellence (Grant no. 10059). Ventura is supported by Vici Grant 639.033.514 of Jan Draisma from the  Netherlands Organisation for Scientific Research. We thank  the anonymous referees for their careful reading of the paper and useful suggestions.

\bibliographystyle{amsalpha}
\bibliography{kronbib}

\end{document}